\documentclass[12pt,letterpaper]{article}
\usepackage{hyperref}
\usepackage{doi}
\usepackage[dvipsnames]{xcolor}
\usepackage[margin=1in,footskip=0.5in]{geometry}
\usepackage[T1]{fontenc}
\usepackage{verbatim}
\usepackage{amsthm}
\usepackage{amsmath}
\usepackage{amssymb}
\usepackage{caption}
\usepackage{natbib}
\usepackage{bbm}
\usepackage{setspace}
\usepackage{subfigure}
\usepackage{xcolor}
\usepackage{graphics}
\usepackage{epsfig}
\usepackage{geometry}
\usepackage{accents}

\usepackage{enumitem}
\usepackage{setspace}
\usepackage{array}
\usepackage[official]{eurosym}
\newcolumntype{x}[1]{>{\centering\arraybackslash\hspace{0pt}}p{#1}}
\usepackage{tikz}
\usetikzlibrary{patterns}
\usetikzlibrary{patterns.meta}
\pgfdeclarepattern{
  name=hatch,
  parameters={\hatchsize,\hatchangle,\hatchlinewidth},
  bottom left={\pgfpoint{-.1pt}{-.1pt}},
  top right={\pgfpoint{\hatchsize+.1pt}{\hatchsize+.1pt}},
  tile size={\pgfpoint{\hatchsize}{\hatchsize}},
  tile transformation={\pgftransformrotate{\hatchangle}},
  code={
    \pgfsetlinewidth{\hatchlinewidth}
    \pgfpathmoveto{\pgfpoint{-.1pt}{-.1pt}}
    \pgfpathlineto{\pgfpoint{\hatchsize+.1pt}{\hatchsize+.1pt}}
    \pgfpathmoveto{\pgfpoint{-.1pt}{\hatchsize+.1pt}}
    \pgfpathlineto{\pgfpoint{\hatchsize+.1pt}{-.1pt}}
    \pgfusepath{stroke}
  }
}

\tikzset{
  hatch size/.store in=\hatchsize,
  hatch angle/.store in=\hatchangle,
  hatch line width/.store in=\hatchlinewidth,
  hatch size=5pt,
  hatch angle=0pt,
  hatch line width=.5pt,
}
\usepackage{pifont}
\usepackage{avant}

\newcolumntype{L}[1]{>{\raggedright\let\newline\\\arraybackslash\hspace{0pt}}m{#1}}
\newcolumntype{C}[1]{>{\centering\let\newline\\\arraybackslash\hspace{0pt}}m{#1}}
\newcolumntype{R}[1]{>{\raggedleft\let\newline\\\arraybackslash\hspace{0pt}}m{#1}}

\makeatletter \renewenvironment{proof}[1][\proofname]
{\par\pushQED{\qed}\normalfont\topsep6\p@\@plus6\p@\relax\trivlist\item[\hskip\labelsep\bfseries#1\@addpunct{.}]\ignorespaces}{\popQED\endtrivlist\@endpefalse} \makeatother

\theoremstyle{plain}

\newtheorem{lem}{Lemma}
\newtheorem{defn}{Definition}
\newtheorem{asmp}{Assumption}

\newtheorem{prop}{Proposition}

\usepackage{multirow}
\usepackage{chngcntr}
\usepackage{apptools}
\AtAppendix{\counterwithin{lem}{section}}
\AtAppendix{\counterwithin{prop}{section}}

\definecolor{uncblue}{RGB}{75, 165, 211}
\definecolor{nberblue}{RGB}{0, 90, 155}
\definecolor{resgreen}{RGB}{34, 84, 67}
\definecolor{pennblue}{RGB}{0, 44, 119}
\definecolor{pennred}{RGB}{152, 30, 50}
\definecolor{beige}{RGB}{210,194,149}

\hypersetup{
pdfborder = {0 0 0},
urlbordercolor = {0 0 0},
colorlinks=true,
linkcolor=pennblue,
urlcolor=nberblue,
citecolor=pennblue}

\newcounter{parentnumber}

\usepackage{float}

\begin{document}
\newcommand{\reals}{\ensuremath{\mathbb{R}}}
\newcommand{\R}{\ensuremath{\mathbb{R}}}
\newcommand{\Z}{\ensuremath{\mathbb{Z}}}
\newcommand{\E}{\ensuremath{\mathbb{E}}}
\newcommand{\expect}{\ensuremath{\mathbb{E}}}
\newcommand{\dint}{\ensuremath{\,\mathrm{d}}}
\newcommand{\demand}{\ensuremath{q}}
\newcommand{\revenue}{\ensuremath{r}}
\newcommand{\control}{\ensuremath{z}}
\newcommand{\qmax}{\ensuremath{q_{max}}}

\newcommand{\fcost}{\ensuremath{k}}

\newcommand{\fcostapp}{\ensuremath{\fcost \cdot \mathbbm{1}_{q(c)>0}}}

\author{Jiaming Wei \and Dihan Zou
}
\date{\today}
\title{Regulating a Monopolist without Subsidy\thanks{
{Wei: Kenan-Flagler Business School, University of North Carolina at Chapel Hill. Email: \href{mailto:Jiaming_Wei@kenan-flagler.unc.edu}{\texttt{Jiaming\_Wei@kenan-flagler.unc.edu.}} Zou: Economics Department, University of North Carolina at Chapel Hill. Email: \href{mailto:dihan@email.unc.edu}{\texttt{dihan@email.unc.edu.}}
We are indebted to Fei Li and Yingni Guo for their invaluable guidance at different stages of this paper.
We thank 
Gary Biglaiser,
Thomas Bollinger,
Joseph Harrington,
David Kim,
Peter Norman,
Mehdi Shadmehr, 
Eran Shmaya,
Francesco Slataper,
Yangbo Song, 
Can Urgun,
Udayan Vaidya,
Peiran Xiao,
Andrew Yates,
Huseyin Yildirim,
Chuan Yu,
Luke Zhao, 
and 
audiences at UNC, Duke, NC State, Stony Brook International Game Theory Conference, Midwest Economic Theory Conference (Penn State), and Theory@Chapel Hill
for their comments. 
All errors are our own.}
}
}

\maketitle

    \vspace{-.3in}

    \begin{abstract}
        We study monopoly regulation under asymmetric information about costs when subsidies are infeasible. 
        A monopolist with privately known marginal cost serves a single product market and sets a price.
        The regulator maximizes a weighted welfare function using unit taxes as sole policy instrument. 
        We identify a sufficient and necessary condition for when laissez-faire is optimal. 
        When intervention is desired, we provide simple sufficient conditions under which the optimal policy is a \emph{progressive price cap}: 
        prices below a benchmark face no tax, while higher prices are taxed at increasing and potentially prohibitive rates.
        This policy combines \textit{delegation} at low prices with \textit{taxation} at high prices, balancing access, affordability, and profitability. Our results clarify when taxes act as complements to subsidies and when they serve only as imperfect substitutes, illuminating how feasible policy instruments shape optimal regulatory design.
\end{abstract}

\begin{description}
{\small
\item[Keywords:] monopoly regulation, no subsidy, delegation, unit tax

\item[JEL Classification Codes:] L43, L5, D82
}

\end{description}

\newpage

\thispagestyle{empty}
\tableofcontents
\addtocounter{page}{-1}

\newpage

\section{Introduction}
Classical models of monopoly regulation assume that regulators can freely deploy the full range of policy instruments including taxes, subsidies, and price caps. In practice, however, these instruments are often restricted. Such restrictions complicate the design of optimal regulation, especially when the regulator lacks full information about production costs. In many environments, subsidies are infeasible due to institutional rules, political externalities, or fiscal constraints. In the European Union, for instance, state aid to firms is tightly restricted, with violations triggering substantial penalties, as in the case \cite{FrancevEC}.\footnote{
France enticed the dominant operator of ferry travel to and from Corsica to offer a richer schedule with a per capita subsidy, but the European Commission considered it inadmissible. The operator, \textit{Soci\'et\'e Nationale Corse-M\'editerran\'ee}, had to repay \euro{220} million of illegal state aid (\cite{legalNotes}). 
}
Subsidizing goods that enter export markets can provoke trade disputes,\footnote{
See \cite{EUTradeWar}. In principle, \cite{WTOsubsidy} prohibits subsidizing exports.
} 
while in developing countries limited public budgets severely constrain governments' scope of subsidy.\footnote{
In many developing countries, governments lack sufficient public funds to subsidize utilities like electricity (\cite{UNReportSubsidy}). 
} These restrictions raise a natural question: How should a monopolist with private cost information be regulated when subsidies are prohibited?

This paper studies monopoly regulation under asymmetric cost information when subsidies are infeasible. We adopt the seminal framework of \cite{BaronMyerson82}: a monopolist serves a single-product market and posts a price, facing downward-sloping demand. 
The firm’s marginal cost is privately observed. The regulator maximizes a weighted welfare function combining consumer surplus and firm profit. Crucially, the regulator cannot use subsidies. Without loss of generality, we focus on unit taxes as the sole policy instrument, which are equivalent to lump-sum excise taxes but provide a natural economic interpretation. A unit tax increases the consumer price relative to the firm’s price, reducing demand. By designing a tax schedule that maps prices into tax rates, the regulator can influence pricing incentives. The policy faces a central tradeoff: absent intervention, the monopolist sets its unregulated monopoly price, generating inefficiently low output; yet while taxes can deter excessive firm prices, they also raise consumer prices and thus depress demand and gains from trade.

We ask two questions: (i) When should the regulator intervene? (ii) When intervention is warranted, what is the optimal policy? 
Our first result provides a necessary and sufficient condition under which the laissez-faire outcome is optimal without regulatory intervention (Proposition \ref{prop_when}). The condition depends on demand curvature, cost distribution, and the regulator’s redistributive preferences. Intuitively, laissez-faire is more likely to be optimal when demand is relatively inelastic, high costs are common, or the regulator places less weight on redistribution. In these environments, taxing high prices would generate limited downward pressure on prices but large distortions in demand, making intervention undesirable.

When intervention is necessary, we find simple sufficient conditions for the optimal policy to be a \textit{progressive price cap}. 
Under this policy, no tax applies to prices below a benchmark level, while prices above this threshold incur a positive tax that is increasing in price, which can be prohibitive for sufficiently high prices.
The benchmark price partitions the price space into two regions: a \textit{delegation region}, where the monopolist retains pricing power, and a \textit{taxation region}, where high prices are deterred. The firm self-selects across these regions based on its cost:
a low-cost firm either chooses laissez-faire prices or the benchmark price, a medium-cost firm pools at the benchmark, and a high-cost firm charges above it, facing taxes and, in extreme cases, exclusion. 
This structure echoes the carrot-and-stick logic of \cite{BaronMyerson82}, 
but without subsidies the regulator cannot reward low prices with transfers;
instead, she offers the weaker reward of price-setting freedom. 
The design of the progressive price cap reflects a three-way tradeoff among \textit{access}, \textit{affordability}, and \textit{profitability}, which are distinct forces in the regulator’s objective. Access concerns how many consumers are ultimately served, capturing the extensive margin of consumer surplus. Affordability refers to the prices paid by those who do purchase, i.e., the intensive margin. Profitability reflects the firm’s ability to cover its costs and remain willing to produce. Improving affordability requires taxing high prices, but doing so makes high-cost firms less profitable and reduces access, since the resulting consumer price moves farther from efficiency. The benchmark price and the accompanying tax schedule in the progressive price cap are chosen to balance these competing objectives.

Our results shed light on the relationship between regulatory instruments. In environments with flexible transfers, the regulator can finely tune incentives, using subsidies to encourage low prices and taxes to deter high ones. By contrast, when subsidies are prohibited, taxes become the sole incentive lever, but their unilateral use generates unbalanced distortions. The absence of subsidies therefore limits the scope for taxation: intervention is warranted only when the welfare gains from discouraging high prices outweigh the distortions created by relying exclusively on taxes.

Relative to the no-transfer setting (where neither taxes nor subsidies are allowed), however, allowing taxes expands the regulator’s ability to shape market outcomes. Taxes provide a more targeted way to discipline prices than coarse tools such as hard price caps, thereby broadening the circumstances in which regulation improves welfare. Taken together, these comparisons illustrate how the availability of subsidies changes the role of taxation: when both instruments are permitted, taxes and subsidies act as complements, jointly shaping incentives across the price spectrum; when subsidies are forbidden, taxes act as an imperfect substitute, shouldering all incentive provision on their own. This perspective clarifies how the structure of feasible policies governs both the scope and the form of optimal intervention.

\paragraph{Related Literature} This paper contributes to the vast literature on monopoly regulation with asymmetric information. We refer to \cite{LaffontTiroleBook93} and \cite{ArmstrongSappington07HIO} for an overview of the field. 
Classical literature begins by modeling different modes of regulator's information disadvantage, including unknown cost (\cite{BaronMyerson82}), unknown demand (\cite{LewisSappington88AER}), and hidden cost-reduction effort (\cite{LaffontTirole86JPE}).
Subsequent literature has been mainly focusing on more severe asymmetric information (\cite{LewisSappington88unknown-demand-and-cost}; \cite{Armstrong99JET})  toward larger uncertainty  
and culminates in robustness
(\cite{GuoShmaya25RobustRegu}; 
\cite{Garrett14GEB};
\cite{MishraPatil2025JET}; \cite{MishraPatilPavan}) and flexibility (\cite{Krahmer25Flexible}).  The mode of asymmetric information in our setting is closest to \cite{BaronMyerson82}. The regulator is uncertain about the firm's marginal cost, but has a prior distribution for policy design.
We extend the literature by introducing policy restrictions alongside informational asymmetry, an aspect that has received comparatively less attention.\footnote{
There is a literature evaluating different policy instruments under specific applications such as pollution abatement (\cite{RobertsSpence1976}, \cite{Baron85RAND}, \cite{Baron85JPubE}).
A recent contribution by \cite{BFV17RANDPriceQuantity} compares price control and quantity control under different modes of information asymmetry.
} Pioneered by \cite{Schmalensee89Good}, the literature on policy restriction focuses on the extreme case of no transfer and highlights the performance of simple instruments (e.g., \cite{BhaskarMcClellanAER23}; \cite{Xia25}).

Closest to our work is \cite{AmadorBagwell22Regulation}, who establish optimality of price cap in general market environments under asymmetric information about production cost without transfers, extending the seminal Lagrangian technique in \cite{AB13ECMA}. 
We bridge the extremes of flexible transfers and no transfer by analyzing an intermediate case: subsidies are prohibited, but taxes remain feasible.  
To our best knowledge, this is the first paper to study optimal monopoly regulation 
under asymmetric information with one-sided transfers.\footnote{
In Section 2.7 of their textbook (pp.\ 151 - 153), \cite{LaffontTiroleBook93} discuss a setting of hidden cost-reduction effort with a budget balance constraint for the firm, and suggest the optimum may be implemented using a ``sliding scale plan.''
}

\section{Environment}\label{section:model}

In a single-product market, a monopolistic firm  produces quantity $q \in [0, \qmax]$ at a constant marginal cost $c$. For $q > 0$, it incurs a fixed cost $\fcost \geq 0$, so its cost function is $C(q) = c\cdot q + \fcost$; otherwise $C(0) = 0$. The marginal cost $c$ is drawn from a twice continuously differentiable distribution $F$ on $[0,1]$, and the realization of $c$, henceforth the firm's type, is privately observed by the firm. The fixed cost $\fcost$ is commonly known, ensuring that asymmetric information pertains solely to the marginal cost relevant for the firm's output decision.\footnote{
This specification also maintains comparability with the canonical models in \cite{BaronMyerson82} and \cite{AmadorBagwell22Regulation}.
} The firm posts a unit price $p$.  

The market demand follows a strictly decreasing and twice continuously differentiable inverse demand function $P\colon [0,\qmax] \rightarrow [0,\overline{v}]$ with the highest willingness-to-pay in the market $\overline{v} > 0$. Given a unit price $p$, the market demand is $q = P^{-1}(p)$, and the resulting consumer surplus is $\int_0^q P(x)\,dx - qP(q)$. 

A regulator maximizes some weighted surplus in the market, but cannot use subsidy. In particular, the only policy instrument that the regulator 
can use is a (nonnegative) unit tax $\tau\colon \R_+ \rightarrow \R_+$, mapping from the firm price to a tax rate.\footnote{
In Lemma \ref{lem:opt_ctrl_reform} of Appendix \ref{apx:main}, we show that it is equivalent to consider lump-sum transfers instead of unit transfer. We use unit tax for the main model to keep the economic interpretation consistent.
}
The regulator has full commitment power on $\tau$.
Given a choice of $\tau$, if the firm price is $p$, then the consumer price is $p + \tau(p)$. 
As a result, the policy generates a \textit{regulated demand function} $q_\tau \colon \R_+ \rightarrow [0,\qmax]$, defined as $q_\tau(p):= P^{-1}(p+\tau(p))$. \
Note that a unit tax $\tau(p) \geq 0$ cannot generate more demand than the market demand without tax, i.e., $q_\tau(p) \leq P^{-1}(p)$. 

We adopt the accounting convention that the firm receives its unit proceeds after the regulator collects the prescribed tax. Faced with the regulated demand $q_\tau(\cdot)$ and privately observing $c$, the firm chooses a price\footnote{
For notational convenience, we assume for now that the firm chooses a (lowest) price in favor of the consumer if there were multiple optimal prices, so that $p(c)$ is a well-defined function. This simplification is not essential since the firm has a unique choice of price under the optimal regulation, as we show later.
} to maximize its profit \begin{equation}
    p(c) = \arg\max_p(p-c)\cdot q_\tau(p) - \fcost\cdot\mathbbm{1}_{q_\tau(p)>0}, \label{primitiveIC}
\end{equation}

One type of regulation commonly seen in the literature (e.g., \cite{BaronMyerson82})  manipulates the demand curve so that it is more sensitive (elastic) to price.
The following example illustrates how different choices of unit tax shape the demand curve.

\paragraph{Example} Suppose the market inverse demand is $P(q) = 1-q$ with $q \in [0,1]$. Consider two policies: \[
\tau_1(p) = \begin{cases}
    0 &\text{ if } p \leq 0.5\\
    0.5 &\text{ if } p > 0.5
\end{cases}\; , \qquad \tau_2(p) = \begin{cases}
    0 &\text{ if } p \leq 0.5\\
    0.5(p-0.5) &\text{ if } p > 0.5
\end{cases}.
\]

\begin{figure}
    \centering
    \begin{tikzpicture}[scale=5.5]
        \footnotesize
            \draw [->, thick,gray] (0,0) -- (0,1.2);
            \draw [->, thick,gray] (0,0) -- (1.2,0);
           
            \node at (-0.07, 1.25) {$p$};
            \node at (1, -0.05) {$1$};
            \node at (-0.05,1) {$1$};
            \node at (1.22, -0.05) {$q$};
            \node at (-0.05, 0.5) {$0.5$};
            \node at (0.5, -0.05) {$0.5$};
            \node at (0.15, 0.95) {$P(q)$};
            \node at (0.1, 0.55) {${\color{pennred} P_1(q)}$};
            
            \draw[-, ultra thick, pennred, opacity = 0.5] (0,1)--(0, 0.5) -- (0.5, 0.5) -- (1,0);

            \draw[-, very thick, dotted] (0,1) -- (1,0);

            \draw [->, thick,gray] (1.5,0) -- (1.5,1.2);
            \draw [->, thick,gray] (1.5,0) -- (2.7,0);
           
            \node at (-0.07+1.5, 1.25) {$p$};
            \node at (1+1.5, -0.05) {$1$};
            \node at (-0.05+1.5,1) {$1$};
            \node at (1.22+1.5, -0.05) {$q$};
            \node at (-0.05+1.5, 0.5) {$0.5$};
            \node at (0.5+1.5, -0.05) {$0.5$};
            \node at (0.15+1.5, 0.95) {$P(q)$};
            
            \draw[-, ultra thick, nberblue, opacity = 0.6] (1.5, 1) -- (1.5, 5/6) -- (2, 0.5) -- (1+1.5,0);
            
            \node at (1.6, 0.65) {${\color{nberblue} P_2(q)}$};
            \draw[-, dashed, gray] (1.5, 0.5) -- (2, 0.5);
            \draw[-, very thick, dotted] (1.5,1) -- (2.5,0);
            \end{tikzpicture}
    \caption{Examples of regulated inverse demand curves. {\footnotesize The left panel compares the regulated inverse demand under $\tau_1$ (red solid line) against the unregulated one (black dotted line); the right panel compares the regulated inverse demand under $\tau_2$ (blue solid line) against the unregulated one (black dotted line)}}
    \label{fig:tax_example}
\end{figure}
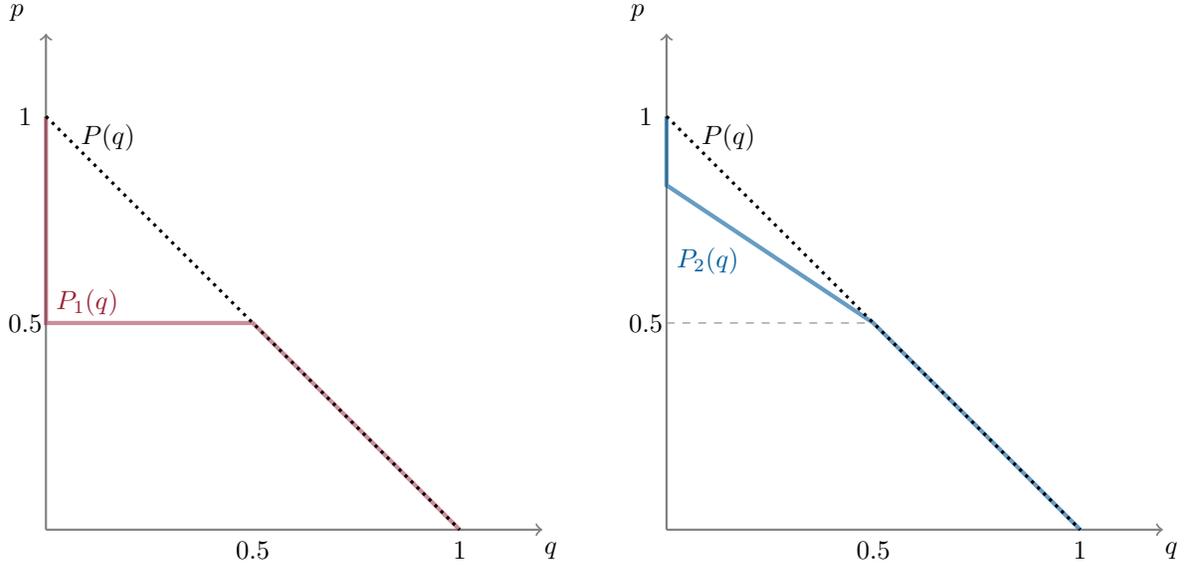
Figure \ref{fig:tax_example} compares the inverse demand curves generated by the two policies. Policy $\tau_1$ is an example of a (hard) price cap. To see this, the regulated demand is \[
q_1(p) 
= \begin{cases}
    1 - p & p \leq 0.5\\
    0 & p > 0.5
\end{cases}.
\]
Posting a price $p > 0.5$ leads to zero demand, while posting a price $p\leq 0.5$ leads to the unregulated market demand. As a result, any type of firm that makes positive profit does not charge a price above $0.5$. In the left panel of Figure \ref{fig:tax_example}, we denote the firm's inverse demand curve corresponding to $q_1$ as $P_1$. This is an extreme case of manipulation the demand curve: the demand is extremely sensitive to the price at one point.

Policy $\tau_2$ is a hybrid of a price cap and progressive tax, leading to a regulated demand of \[
q_2(p)
= \begin{cases}
    1 - p & p \leq 0.5\\
    1.25 - 1.5p & p \in (0.5, \frac{5}{6}]\\
    0 & p > \frac{5}{6}
\end{cases}.
\]
In the right panel of Figure \ref{fig:tax_example}, we denote the firm's inverse demand curve corresponding to $q_2$ as $P_2$.
If the firm charges a price $p \leq 0.5$, it will enjoy the unregulated market demand according to $P(q)$; if the firm charges $p > 0.5$, leading to $q < 0.5$, it will face a segment that is more sensitive to price change than the original market demand.

\medskip

\paragraph{Timing} The timing of the model is summarized as follows:
\begin{itemize}
    \item[1.] The regulator commits to a unit tax policy $\tau$.
    \item[2.] Nature draws the firm's cost $c$ according to $F$.
    \item[3.] The firm observing its cost $c$ and the policy $\tau$, posts a price $p$.
    \item[4.] Consumers decide whether to purchase, after which transfers and payoffs are realized.
\end{itemize}

\paragraph{Regulator's Objective} The regulator maximizes the weighted surplus. Let the welfare weight on firm profit ($\Pi$) to be $\alpha \in [0,1]$, and the welfare weight on consumer surplus (CS) be $1$, then 
the regulator maximizes $\text{CS} + \alpha \Pi$.
Since the welfare weight on consumer is weakly larger, it is without loss to assume that the regulator rebates the tax revenue to the consumers. Given each price $p$ chosen by the firm with type $c$, and denote the firm's value function to be  $\Pi_\tau(c) := \max_x (x-c)q_\tau(x)  - \fcost\cdot\mathbbm{1}_{q_\tau(x)>0}$, the weighted surplus is 
\begin{align*}
    &\underbrace{\int_{0}^{q_\tau(p)} P(x)\, dx}_{\text{total surplus}} - \underbrace{[p+\tau(p)]q_\tau(p)}_{\text{transfers from consumers}} + \underbrace{\tau(p)q_\tau(p)}_{\text{rebate to consumers}}  +\quad \alpha \cdot \Pi_\tau(c)  \\
    &= \underbrace{\int_{0}^{q_\tau(p)} P(x)\, dx - cq_\tau(p)  - \fcost \cdot \mathbbm{1}_{q_\tau(p) >0} }_{\text{trading surplus}} \quad -\quad  (1-\alpha)\cdot \Pi_\tau(c).
\end{align*}
 There are two observations on the above transformation. To begin with, the weighted surplus is the trading surplus less a fraction of firm's profit. When $\alpha = 0$, the regulator has the maximum redistributive motive and only cares about consumer surplus, penalizing any profit one-for-one in the objective. When $\alpha = 1$, the penalizing term vanishes, resulting in a utilitarian objective that is the same as the trading surplus. The cases with $\alpha \in (0,1)$ captures regulator's relative preference of consumer surplus protection versus efficient trade. 

The other observation is that $\tau(p) = P(q_\tau(p)) - p \geq 0$, according to the definition of the regulated demand $q_\tau$. Because $\tau$ and $q_\tau$ are one-to-one, it is equivalent for the regulator to choose $q_\tau$ instead of $\tau$.
Denote the consumer value as $V(q):=  \int_{0}^q P(x) dx$. Recall that the firm's pricing strategy $p(c)$ is defined by \eqref{primitiveIC}. The regulator solves the following problem:

\begin{equation}
      \max_{q_\tau \colon \mathbb{R_{+}} \rightarrow [0,\qmax]}\int_0^1 \left[V(q_\tau(p(c))) -cq_\tau(p(c)) -  \fcost \cdot \mathbbm{1}_{q_\tau(p) >0}  - (1-\alpha)\cdot \Pi_\tau(c)\right]dF(c), \label{primitiveOBJ}
  \end{equation}
subject to firm's optimal pricing \eqref{primitiveIC}, firm making nonnegative profit\begin{equation}
    \Pi_\tau(c) \geq 0, \forall c \in [0,1], \label{primitiveIR}
\end{equation}
and feasibility for regulated demand \begin{equation}
    q_\tau(p) \leq P^{-1}(p), \forall c \in [0,1] .\label{primitiveFS}
\end{equation}
 Note that condition \eqref{primitiveFS} is equivalent to $\tau(p(c)) \geq 0$ for all $c\in [0,1]$, and is satisfied by previous policy examples. If subsidy were allowed ($\tau < 0$), the regulated demand would exceed the unregulated one, i.e., firm's inverse demand curve would be above $P(q)$ at least for some $q$.\footnote{
 For the examples in Figure \ref{fig:tax_example}, the firm's inverse demand curves are always \textit{below} the market demand curve $P(q)$.
 }

 \subsection{Preliminary analysis}

 Following the revelation principle, we can focus on the direct mechanism in which the firm reports its type truthfully to the regulator, and then the regulator recommends the price and chooses the demand. 
 Slightly abusing the notation, define the demand policy as $\demand\colon [0,1] \rightarrow [0,\qmax]$ and the pricing policy as $p \colon [0,1] \rightarrow \R_+$, both mapping from the type space. Denote the value function of the firm as  $\Pi\colon [0,1] \rightarrow \R_+$, given by \[
 \Pi(c) := \max_x (p(x)-c)q(x)  - \fcost\cdot\mathbbm{1}_{q(x)>0}
 \]
 Truth-reporting requires the firm to optimally report $x = c$ in the above problem, which leads to $\Pi(c) = (p(c)-c)q(c) + \fcost\cdot\mathbbm{1}_{q(c)>0}, \forall c \in [0,1]$.
 Therefore, the regulator's problem is \begin{equation}
      \max_{p,\demand}\int_0^1 \left(V(\demand(c)) - cq(c) -  \fcost \cdot \mathbbm{1}_{q(c) >0}  - (1-\alpha)\Pi(c) \right)dF(c)\label{directOBJ}
  \end{equation}
  subject to, for all $c\in [0,1]$,
  \begin{align}
      c &\in \arg \max_x (p(x) - c) \demand(x) - \fcost\cdot\mathbbm{1}_{q(x)>0} 
      \tag{IC} \label{directIC}\\
      0 &\leq \Pi(c)\tag{IR}\label{directIR}\\
      \demand(c) &\leq P^{-1}(p(c)). \tag{NS} \label{directFS}
  \end{align}
Problem \eqref{directOBJ} is standard except for the unconventional no-subsidy constraint \eqref{directFS}, which is a type-by-type policy constraint. This constraint fundamentally alters the optimal solution relative to the conventional regulation literature.
To gain intuition about its role, consider a multiplier function $\lambda(c)\ge 0$ associated with \eqref{directFS} for each $c\in[0,1]$, interpreted as the shadow cost of tightening the constraint.
A simplified version of the optimal condition for Problem \eqref{directOBJ} can be written as
\begin{equation}
    P(q(c)) = c \quad + \quad \underbrace{(1-\alpha)\cdot\frac{F(c)}{f(c)}}_{\text{monetary distortion}} + \underbrace{\frac{\int_0^c\lambda(x)\, dx}{f(c)}}_{\text{allocative distortion}}. \tag{MB-MC}\label{pseudoFOC}
\end{equation}
The condition describes the marginal tradeoff for a generic type $c$:
 increasing production by one unit yields the marginal benefit (the value of the marginal consumer served) on the left-hand side 
 and
 incurs the marginal cost on the right-hand side. The marginal cost can be decomposed into three components.
 
 \vspace{.1in}\noindent{\it Production cost.} The first term, $c$, is the firm's marginal {production cost} and is present even under first-best setting with complete information. 

 \vspace{.1in}\noindent{\it Monetary distortion (``information rent'').}
 With private information, 
the regulator must compensate stronger types (those with cost below $c$) to maintain truthful reporting. 
Requiring type $c$ to produce one more unit forces the regulator to give each stronger type an additional dollar.
Because the firm's profit function is quasi-linear whenever output is positive, these transfers do not distort the allocation through the firm's decision. Instead, they enter one-for-one as a marginal cost to the regulator. The factor $(1-\alpha)$ reflects how much the regulator dislikes transfers to the firm. When $\alpha = 1$ (utilitarian), the regulator does not internalize these transfers into marginal cost.

 \vspace{.1in}\noindent{\it Allocative distortion.} The no-subsidy constraint \eqref{directFS}, however, introduces an additional allocative distortion captured by the third term. Unlike monetary distortion, this distortion arises even when $\alpha=1$. Although the regulator would like to compensate stronger types monetarily, doing so tightens their no-subsidy constraint. When \eqref{directFS} binds, the only remaining way to deliver compensation is to distort the allocation to generate more producer surplus, while reducing total trading surplus. The shadow cost $\lambda(x)$ measures, for each type $x<c$, the cost of such distortions, which the regulator must internalize in the marginal tradeoff.

 This new dichotomy between monetary and allocative distortions makes the optimal regulation markedly different from the classical results. Traditional monopoly regulation with flexible transfers (e.g., \cite{BaronMyerson82}) incentivizes low prices by granting additional demand (with subsidies) to low-cost reports and reducing demand (with taxes) for high-cost reports. The resulting regulation is distortionary in demand relative to ex post efficiency (marginal cost pricing: $P(q(c)) = c$).

In contrast, subsidies are prohibited by \eqref{directFS}, and the allocative distortion becomes an inherent part of the marginal cost. We show that in some market environments no intervention is optimal (Proposition \ref{prop_when}), particularly when monetary distortion is small and allocative distortion is large. We also identify simple sufficient conditions under which the conventional logic still applies (Proposition \ref{prop_main}), with one modification. For high-cost reports, the regulator continues to impose a tax, reducing demand relative to the unregulated market; in this region the no-subsidy constraint \eqref{directFS} is slack. For low-cost reports, however, the regulator cannot provide subsidies or expand demand beyond market demand, so constraint \eqref{directFS} must bind. Under the sufficient conditions, \eqref{directFS} becomes single-crossing: there exists a cost cutoff determining where the constraint binds.

  \bigskip

   Another important observation from Problem \eqref{directOBJ} is that the policy instrument of unit tax is equivalent to a lump-sum excise tax. To see this, note that the regulator can design $(q, \Pi)$ instead of $(q, p)$. For example, whenever $q(c) = 0$, constraint \eqref{directFS} trivially holds. Whenever $q(c) > 0$, constraint \eqref{directFS} is equivalent to $P(q(c)) \geq p(c) = (\Pi(c) + cq(c) + \fcost)/q(c)$, or \[
   q(c)\cdot [P(q(c)) - c] - \fcost \geq \Pi(c).
   \]
   Any difference between the left-hand side (unregulated producer surplus) and the right-hand side (firm's profit) is the lump-sum tax collected by the government. We show the equivalence of regulation formally in Appendix \ref{apx:main}.

\medskip

Before diving into the main results, we state our assumption on the demand curve.
\begin{asmp}\label{assumption_concave}
    (i) The revenue function $qP(q)$ is strictly concave in $q$. (ii) $\lim_{q\rightarrow 0}P(q) = \overline{v}$ and $\lim_{q\rightarrow q_{max}}P(q) = 0$. (iii) There exists $q \in [0, q_{max}]\colon qP(q) > k$. 
\end{asmp}

To illustrate the rationale behind this assumption, 
we study a benchmark case with no intervention, i.e., $\tau(p) \equiv 0$; or equivalently, $\demand(c) = P^{-1}(p(c))$ for all $c$. We call this policy \textit{laissez-faire}. Given Assumption \ref{assumption_concave}, the firm's strategy is well-behaved under laissez-faire policy: conditions (i) and (ii) guarantee that the optimal production is characterized by the first-order condition, is unique, and is monotone in $c$. Condition (iii) bounds the fixed cost from above so that it is profitable for the lowest-cost firm ($c=0$) to produce. 
\begin{lem}[Laissez-Faire Baseline]
    \label{lem:lf} There is a unique laissez-faire production strategy\[
    q_{LF}(c) = \begin{cases}
        \hat{q}(c), &c\leq \overline{c}_{LF}\\
        0, &c> \overline{c}_{LF}
    \end{cases}
    \] for some cutoff $\overline{c}_{LF} \in (0, 1]$, where $\hat{q}(c)$ solves $P(q) = c - qP'(q)$ for each $c$.
  Furthermore, $q_{LF}(c)$ strictly decreases in $c$ for $[0, \overline{c}_{LF})$.
\end{lem} 
All proofs are in the \hyperref[appendix]{Appendix}. 
Lemma \ref{lem:lf} implies that, under the laissez-faire policy, the firm with higher marginal cost $(c > \overline{c}_{LF})$ does not produce because it cannot recover the fixed cost. The firm with lower marginal cost $(c < \overline{c}_{LF})$ produces and chooses a unique pricing strategy $P(q_{LF}(\cdot))$ equal to its marginal cost plus some markup. $P(q_{LF}(\cdot))$ will frequently appear as a baseline for comparison with the optimal regulation.

\subsection{Discussion}

We restrict the regulator to using a unit tax, ruling out piece-rate or lump-sum subsidies. As shown in Lemma \ref{lem:opt_ctrl_reform}, an excise (lump-sum) tax is equivalent to a unit tax in our environment. Without the transfer constraint, the regulator could employ a fully flexible policy combining both taxes and subsidies to manipulate the price sensitivity of the demand curve and intervene at virtually all price levels. Our formulation is therefore equivalent to \cite{BaronMyerson82} with an additional no-subsidy constraint, and also equivalent to \cite{AmadorBagwell22Regulation} once their ``no-tax’’ restriction is removed (see Lemma \ref{lem:opt_ctrl_reform}).

An alternative interpretation of the model is as a procurement problem. The regulator purchases output from the monopolist and either consumes it directly (e.g., national defense) or \emph{assigns} it to consumers (e.g., healthcare services). In both interpretations, transfers outside of the transaction are not allowed. Under this view, the unit tax corresponds to demand rationing, and the no-subsidy constraint \eqref{directFS} requires that procurement not exceed market demand.

Finally, we assume that the revenue function induced by market demand is strictly concave, that the demand curve reaches both its upper and lower bounds, and that the distribution of cost types is twice continuously differentiable. These assumptions are stronger than those in \cite{BaronMyerson82} and weaker than those in \cite{AmadorBagwell22Regulation}. We regard them as mild: they still accommodate a wide range of market primitives while ensuring a well-behaved laissez-faire pricing schedule. They also deliver continuity, existence, and uniqueness of the optimal regulatory policy. In Lemma \ref{lem:existence}, we show that it is without loss of generality to restrict attention to deterministic policies. 

\section{Optimal Regulation}\label{sec:opt_R}
In this section, we first report a sufficient and necessary condition for the laissez-faire economy without intervention to be optimal. When the condition fails, we characterize the optimal policy combining delegation and taxation.

Roughly speaking, the condition that warrants no intervention is more likely to hold if the policy environment has one or more of the following features: (a) the demand curve is relatively insensitive to price, so that the monopoly markup increases with cost; (b) the density of cost distribution concentrates more on high costs;  or (c) the regulator cares less about redistribution.

When the condition fails, we find simple sufficient conditions for the optimal intervention to be a \textit{progressive price cap}: The regulator sets a \textit{benchmark price} $\hat{p} \geq 0$. If the firm prices below $\hat{p}$, the regulator imposes no tax ($\tau(p) = 0$), and the regulated demand coincides with the market demand, i.e., $q_\tau(p) = P^{-1}(p)$. If the firm prices above $\hat{p}$, the regulator imposes a positive tax $\tau(p) > 0$ that is increasing in $p$, and the regulated demand is strictly less than the laissez-faire demand. Different types of the firm self-select into the two policy regions, respectively.

\subsection{Laissez-faire as the optimal regulation}\label{sec:lf}

The following proposition reports the sufficient and necessary condition for the laissez-faire policy to be optimal. 

\begin{prop}\label{prop_when}
    It is optimal to use the laissez-faire policy $(q_{LF}(c), P(q_{LF}(c)))$ if and only if \begin{equation}
        [P(q_{LF}(c)) - c]f(c) - (1-\alpha)F(c) \text{ is nondecreasing in } c, \forall c \in [0,\overline{c}_{LF}]. \label{lf_optimal} \tag{LF}
    \end{equation}
\end{prop}
To build intuition for condition \eqref{lf_optimal}, note that \eqref{pseudoFOC} boils down to whether one can construct a set of positive multipliers $\lambda(c)$ for all $c$ in the optimization.
Substitute the laissez-faire allocation into \eqref{pseudoFOC} and rearrange to obtain \[
\int_0^c \lambda(x)\, dx = [P(q_{LF}(c)) - c]f(c) - (1-\alpha)F(c).\] 
This expression can be interpreted as positing a nondecreasing cumulative multiplier consistent with laissez-faire. If this proposed multiplier is valid (i.e., if the right-hand side is nondecreasing in $c$), then the laissez-faire allocation satisfies the optimality conditions.

\paragraph{Economic Interpretation} Consider the extreme case of $\alpha = 1$, where the regulator is fully utilitarian and maximizes total trading surplus. Condition \eqref{lf_optimal} requires $[P(q_{LF}(c)) - c]f(c)$ to be nondecreasing in $c$, implying that the monopoly markup, $P(q_{LF}(c)) - c$, and the density of cost types, $f(c)$, cannot decline too rapidly. Hence, \eqref{lf_optimal} is satisfied when (a) demand is relatively inelastic to price (e.g., when consumers are more homogeneous), and (b) high-cost types occur with sufficiently high probability. 

The roles of (a) and (b) can be understood by comparing the laissez-faire inefficiencies versus those under regulation. Under laissez-faire, inefficiency arises because the monopolist sets a high price, excluding consumers whose value exceeds the production cost. 
A regulator could mitigate this by deterring excessive price-setting through taxation. 
Such intervention, however, lowers inefficiency for intermediate-cost firms while raising it for high-cost firms.
Under (a), laissez-faire inefficiency is already small (in the limit of homogeneous consumers, the monopolist cannot profitably exclude any buyers).  Regulation would thus create larger distortions than it removes. Under (b), if high-cost types are sufficiently common, excluding them leaves many consumers underserved, again making regulatory intervention more harmful than laissez-faire.

When $\alpha \in [0,1)$, the regulator places extra weight on consumer surplus. Condition \eqref{lf_optimal} can then be interpreted as a weighted monopoly markup adjusted by a redistributive term, $-(1-\alpha)F(c)$. 
As the regulator becomes more consumer-oriented (i.e., as $\alpha$ decreases), the desire to reduce prices strengthens, and it becomes harder to justify laissez-faire.\footnote{%
In footnote 32, \cite{AmadorBagwell22Regulation} (p.\ 1743) note that their exclusion threshold weakly increases in $\alpha$.}

\paragraph{Relation to Alternative Regulatory Environments} Condition \eqref{lf_optimal} is complementary to the sufficient conditions for nonintervention in \cite{AmadorBagwell22Regulation}, Proposition 3.\footnote{%
Their sufficient conditions are: (a') $-qP'(q)$ is decreasing in $q$, and (b) $f(c)$ is nondecreasing. Condition (a') implies the monopoly markup $-q_{LF}(c)P'(q_{LF}(c))$ decreases in $q_{LF}(c)$, equivalent to part (a) in \eqref{lf_optimal} increasing in $c$. When $\alpha = 1$, their conditions imply ours. For general $\alpha \in [0,1)$, the two sets of conditions are complementary.}
The regulatory environment they study—monopoly regulation without transfers—is equivalent to ours with an additional no-tax constraint. There, the regulator can use only a price cap to exclude high-cost types. This instrument can be replicated through taxation in our framework (cf.\ policy example $\tau_1$).  Naturally, if it is optimal not to intervene via taxes (i.e., without subsidies) in a given market environment, then it is also optimal not to intervene via a price cap (i.e., without transfers). Thus, \eqref{lf_optimal} serves as a sufficient condition for laissez-faire to be optimal when transfers are prohibited. Importantly, allowing taxes expands the regulator’s intervention set relative to the no-transfer case.
The underlying tradeoff is similar, though: reducing prices versus distorting high-cost types.

At the opposite extreme, when the regulator can use fully flexible transfers, \cite{BaronMyerson82} show that intervention is typically desirable. From a feasibility perspective, laissez-faire is always feasible regardless of transfer restrictions. With flexible transfers, however, laissez-faire is sub-optimal except for knife-edge cases. Introducing the no-subsidy constraint reduces the feasible set in a way that can make nonintervention optimal, as shown in Proposition \ref{prop_when}. From the perspective of policy instruments, Proposition \ref{prop_when} also shows that prohibiting subsidies limits the regulator’s ability to use taxes effectively.

\paragraph{Taxes and Subsidies: Substitute and Complementary Instruments} The preceding discussion suggests that taxes and subsidies function both as substitute and as complementary instruments. They are substitute in that either can be used to induce lower prices, albeit at the cost of distortion. But they are complementary because the distortions they induce move in different directions: taxes reduce output, while subsidies entail fiscal costs (when $\alpha < 1$). When both tools are available, the regulator balances these distortions—using taxes and subsidies jointly to maintain low prices without excessively depressing output or incurring excessive expenditure. This complementarity reduces overall welfare loss. In contrast, when transfers are infeasible, the absence of the complementary instruments often renders the welfare loss too large to justify intervention (\cite{AmadorBagwell22Regulation}, Proposition 3).

On the other hand, when only taxes are available, the regulator must lower prices solely by inducing underproduction, which may as well outweigh the gains from intervention (in environments satisfying \eqref{lf_optimal}). When the intervention is nevertheless beneficial, it remains optimal to use taxes alone, reflecting the substitutability of the instruments. This is the case analyzed in Section \ref{sec:intervention}.

\bigskip

Next, we present examples of market environments where laissez-faire is or is not an optimal regulation policy. In particular, when the inverse demand curve is linear, consumer valuations are more diverse. The demand elasticity increases sharply as $q$ decreases, and the regulator often finds it optimal to intervene. When the inverse demand curve is constant-elastic, on the other hand, consumer valuations are more homogeneous, and there are many cases for nonintervention to be optimal. 

\paragraph{Linear Demand} We show that the linear demand $P(q) = A - Bq$ (with $A \leq 2B$) violates condition \eqref{lf_optimal} for many choices of density function $f$ and welfare weight $\alpha$. 
By Lemma \ref{lem:lf}, $q_{LF}(c) = \max\{\frac{A-c}{2B}, 0\}$, and  the expression in condition \eqref{lf_optimal} becomes\footnote{
We use the shorthand notation $[x]^+ := \max\{x,0\}$. 
} \[
[P(q_{LF}(c)) - c]f(c) - (1-\alpha)F(c)= \frac{[A-c]^+}{2}f(c) - (1-\alpha)F(c).
\]
Note that $[A - c]^+$ is a nonnegative decreasing function in $c$. If $f(c)$ is nonincreasing (e.g., uniform and truncated exponential), the whole expression is nonincreasing. It is also easy to check that the expression strictly decreases in some intervals for truncated normal densities. We conclude that there is room of intervention under linear demand curve and many common cost distributions, regardless the welfare weight $\alpha$.
In contrast, without any transfer, laissez-faire will be optimal under the environment of linear demand, uniform cost distribution, and $\alpha =1$.\footnote{
See \cite{AmadorBagwell22Regulation}'s result without fixed cost (their footnote 33, p.\ 1744).
}
For a concrete example on the substitution relationship between taxes and subsidies, we compare the optimal policy in Section \ref{sec:intervention} 
for this particular environment across three policy sets: 
only taxation, no transfer, and flexible transfer.\footnote{
We henceforth refer to \cite{BaronMyerson82} whenever discussing on the optimal policies with flexible transfer, and refer to \cite{AmadorBagwell22Regulation} whenever discussing on the optimal policies with no transfer.
}

\paragraph{Constant-Elastic Demand} Consider the inverse demand function $P(q) = \theta q^{-\frac{1}{\eta}}$, where $\theta > 0$ is a small constant, and the constant $\eta > 1$ is (absolute value of) the demand elasticity. In contrast to linear demand, many environments with constant-elastic demand justify laissez-faire policy.\footnote{
Because we require the set of demand to be compact and that the maximum willingness-to-pay is $\overline{v}$, we need to truncate the demand curve slightly to $P^\epsilon(q) = \max\{\theta q^{-\frac{1}{\eta}} - \epsilon, \overline{v}\}$, for some small $\epsilon>0$. 
Fix a small $\theta>0$, we can find a tiny $\epsilon > 0$ so that $P^\epsilon(q)$ has approximately constant elasticity. With this modification, the fixed cost $\fcost$ needs to be reasonably large to guarantee the markup is monotone.
Then, the discussion in the text goes through. For exposition clarity, however, we suppress the exact treatment here. 
}
By Lemma \ref{lem:lf}, \[
P(q_{LF}(c)) - c = \frac{c}{1 - \frac{1}{\eta}}-c= \frac{c}{\eta - 1},
\]
which increases in $c$. Therefore, condition \eqref{lf_optimal} requires \[
\frac{cf(c)}{\eta - 1} - (1-\alpha)F(c)
\]
to be nondecreasing. A sufficient condition is that $f(c)$ is nondecreasing and $\alpha$ is large. For example, when $f$ is uniform,\footnote{
Many cases with truncated normal or truncated exponential distribution also satisfy condition \eqref{lf_optimal}. 
} we need \[
\frac{1}{\eta - 1} \geq  1-\alpha,
\]
which holds if $\alpha$ is not too small (regulator cares about firm profit) and $\eta$ is not too large (the profit margin is large enough).

\subsection{Intervention as the optimal regulation} \label{sec:intervention}
We now focus on market primitives under which condition \eqref{lf_optimal} fails, implying that taxation yields a strict welfare improvement relative to laissez-faire. To analyze the structure of optimal intervention, we introduce a class of tax policies with a simple carrot-and-stick form: a \emph{progressive price cap}.

\begin{defn}[Progressive Price Cap]
    A unit tax policy $\tau_{\hat{p}}(p)$ with a benchmark price $\hat{p} > 0$ is a \emph{progressive price cap} if\begin{itemize}
        \item[(i)] $\tau_{\hat{p}}(p) = 0$ for all $p \leq \hat{p}$;
        \item[(ii)] $\tau_{\hat{p}}(p) > 0$ for all $p > \hat{p}$, with $\tau_{\hat{p}}(p)$ strictly increasing in $p$; and 
        \item[(iii)] there exists some $p > \hat{p}$ such that $P^{-1}(p+\tau_{\hat{p}}(p)) > 0$.
    \end{itemize} 
\end{defn}
Condition (i) defines a \emph{delegation region}: the firm may freely set any price up to $\hat{p}$.
Condition (ii) defines a \emph{taxation region}, where prices above $\hat{p}$ face progressively higher \emph{unit} taxes.\footnote{%
In contrast, lump-sum taxes generally do not have the progressive feature.
}
Together, (i)–(ii) describe a progressive tax implementation of the benchmark price $\hat{p}$.
Condition (iii) requires the price cap to be \emph{soft}: taxes are not prohibitive, and for some $p>\hat{p}$ positive demand remains.
In the policy examples in Section \ref{section:model}, $\tau_2$ is a progressive price cap, while $\tau_1$ is a ``hard’’ price cap.

We next provide simple sufficient conditions under which the optimal intervention takes the form of a progressive price cap and discuss its economic interpretation.

\begin{prop} \label{prop_main}
Suppose condition \eqref{lf_optimal} is violated, $f$ is nonincreasing and log-concave, and the demand function $P^{-1}$ is strictly log-concave.\footnote{%
These conditions complement those in \cite{AmadorBagwell22Regulation} (Proposition 1, p.\ 1735) and hold for a broad class of demand and cost specifications.
For example, if demand is microfounded by a continuum of consumers with i.i.d.\ valuations drawn from $G$, strict log-concavity of the density $g$ implies $P^{-1}(p) = 1 - G(p)$ is strictly log-concave. Also, log-concavity of $P^{-1}$ is stronger than the strict concavity of the revenue function $qP(q)$, so our Assumption \ref{assumption_concave} remains valid (see \cite{ZouLogConcave}).
}
Then there exists $\overline{k}>0$ such that for all $k \in [0,\overline{k})$, the optimal regulation is a progressive price cap $\tau_{\hat{p}}(p)$. Moreover, there exists a cutoff $\hat{c}\in(0,\overline{c}_{LF})$ such that firms with marginal cost above (resp. below) $\hat{c}$ choose prices above (resp. below) $\hat{p}$.
\end{prop}
The proposition characterizes the optimal regulatory instrument, the induced pricing behavior, and the resulting allocation.
When intervention is warranted, the optimal policy consists of both delegation and taxation. Higher prices are discouraged through increasing tax rates, and firms self-select into the two regions according to cost.
High-cost types ($c > \hat{c}$) set $p > \hat{p}$ and are taxed; sufficiently high-cost types may face prohibitive taxes and be excluded.
Low-cost types ($c \le \hat{c}$) set $p = \hat{p}$ and are untaxed;\footnote{
It is well-known in the mechanism design literature that constraints on transfers can lead to pooling allocation. See, e.g., \cite{KrishnaMorgan08}, \cite{PaiVohra14}.
} those with sufficiently low costs may instead optimally charge their laissez-faire price $P(q_{LF}(c)) < \hat{p}$.
Lemma \ref{lem:main} reports explicit expressions for firms’ optimal pricing and the resulting tax schedule.

To understand the sufficiency conditions, note that progressivity requires $\tau(p)$ to increase in $p$. Since the optimal firm price $p(c)$ is increasing in $c$, this is equivalent to requiring $\tau(p(c)) := P(q(c)) - p(c)$ to be increasing in $c$, from the perspective of the direct mechanism. A sufficient condition is that the consumer price markup $P(q(c)) - c$ is increasing while the firm’s price markup $p(c) - c$ is decreasing.
Nonincreasing $f$ ensures that the optimality condition \eqref{pseudoFOC} yields an increasing $P(q(c)) - c$, while log-concavity of $f$ and $P^{-1}$ ensures that $p(c) - c$ is decreasing. The small fixed cost requirement reflects that $\fcost$ must be recovered through unit price without subsidy. Because higher-cost types produce less, the average fixed cost term $\frac{\fcost}{q(c)}$ increases in $c$, pushing $p(c) - c$ upward and counteracting the desired decrease in the firm's markup. A small $\fcost$ ensures this effect remains dominated.

\paragraph{Economic Interpretation} The economics of the progressive price cap can be understood by examining the regulator’s objective. Consider first the extreme case $\alpha = 0$, where the regulator cares solely about consumer surplus. Consumer surplus increases when (a) more consumers gain \emph{access} to the product, and (b) the product becomes more \emph{affordable}.

To improve affordability without subsidies, the regulator must tax high prices, which exerts downward pressure on firm prices. If the firm attempts to set a high price, the tax inflates the consumer price further, reducing demand and making high-price strategies unprofitable.

However, taxation is a double-edged sword. For high-cost firms, a higher consumer price and lower demand move the allocation further from efficiency (where consumer price equals marginal cost), thereby harming consumer access—objective (a). Thus, the optimal tax balances these two forces.

Now consider $\alpha > 0$, so the regulator also cares about (c) firm profit. Distortionary taxation now reduces not only consumer surplus but also the firm’s gains from trade. Regulation must therefore navigate a three-way tension among access, affordability, and profitability.\footnote{%
As emphasized by \cite{GuoShmaya25RobustRegu}, the regulator must balance: mitigating underproduction (\emph{access} and \emph{profitability}), protecting consumer surplus (\emph{affordability}), and limiting overproduction (irrelevant under no-subsidy constraints).
}

The two policy regions generate inefficiencies through different channels.
The taxation region implements demand rationing and sometimes exclusion, reducing both consumer access and firm profit.
The delegation region preserves all cost-efficient trades but at elevated prices, reducing consumer surplus per transaction.

 \paragraph{Relation to Alternative Regulatory Environments} Compared with the flexible-transfer environment, the optimal policy in Proposition \ref{prop_main} also induces a more elastic demand curve, albeit more coarsely.
With flexible transfers, the regulator subsidizes low prices and taxes high prices.
Without subsidies, the regulator’s best tool for rewarding low-price behavior is to grant \emph{delegation} (i.e., the firm’s right to set the price freely) up to the benchmark $\hat{p}$. This delegation region corresponds to the binding no-subsidy constraint \eqref{directFS}. Consequently, low-cost firms either charge $\hat{p}$ or, if sufficiently efficient, revert to their laissez-faire price.

By contrast, in a pure no-transfer environment, the optimal intervention takes the form of a hard price cap. Under a hard cap, firms either charge their laissez-faire price or the cap itself, whichever is lower. Each type of the firm therefore either fully exploits or completely mutes its private information.

The progressive price cap provides more flexibility. In addition to delegation, it allows the firm to \emph{negotiate upward} from $\hat{p}$ via higher prices  at the cost of taxation. This mechanism elicits private information in a way that improves welfare. The benchmark price $\hat{p}$ therefore separates two uses of private information: one for profit maximization in the laissez-faire region and one for welfare maximization under taxation. The relative benefits of these uses, together with the cost of information elicitation, determine the optimal level of $\hat{p}$.

\medskip

In what follows, we present a numerical example to illustrate how the optimal policy operates. We focus on a linear demand curve\footnote{
Appendix \ref{apx:linear-uniform} presents a full characterization of the environment with linear demand curve $P(q) = A - Bq$ and a uniform cost distribution. The requirement that $P^{-1}$ be strictly log-concave also covers other common cases, such as logarithmic inverse demand $P(q) = \mu-\beta\log(q)$.
} $P(q) = 1 - q$ with $q \in [0,1]$, and consider a uniform cost distribution. We also compare the optimal policy across three policy sets: taxation only, flexible transfers, and no transfer.

\paragraph{Uniform Cost} Suppose marginal cost $c \in [0,1]$ is uniformly distributed and the fixed cost is $\fcost = 0$.
Fix $\alpha = 1$.

This example illustrates the pattern described in Proposition \ref{prop_main}: low-cost types pool at the benchmark price, medium-cost types select prices above the benchmark and face a progressive (but non-prohibitive) tax, and sufficiently high-cost types face a prohibitive tax and are excluded—even without fixed cost. Figure \ref{fig:uniform} plots the optimal policy and the corresponding pricing strategy. Closed-form expressions for the optimal regulation and market allocation are provided in Appendix \ref{apx:linear-uniform}.

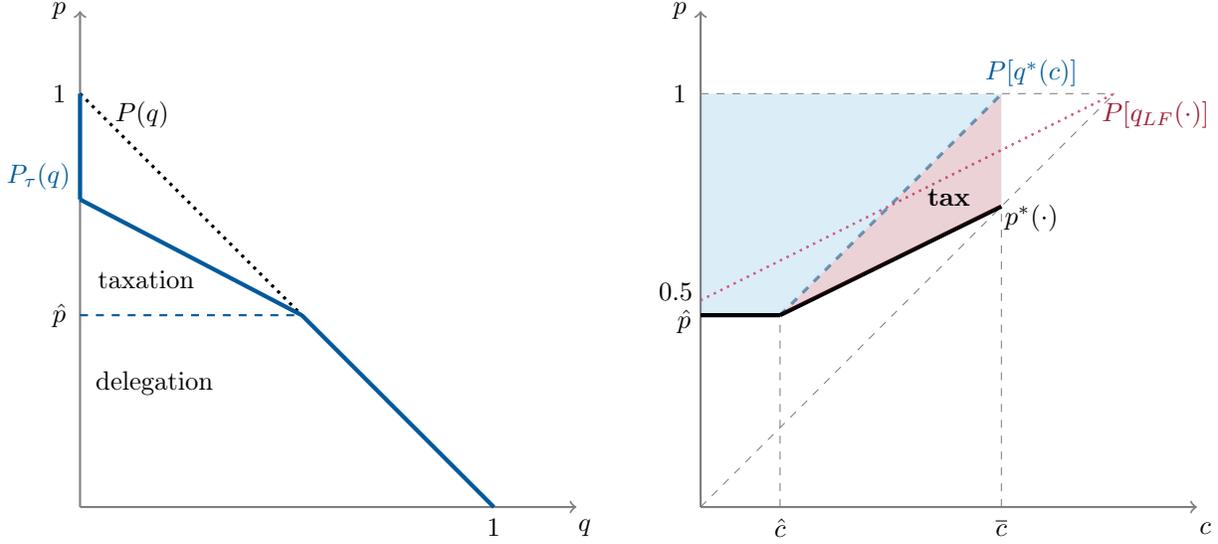
\begin{figure}[htbp]
    \begin{center}
               \begin{tikzpicture}[scale=5.5]
        \footnotesize
            \draw [->, thick,gray] (0,0) -- (0,1.2);
            \draw [->, thick,gray] (0,0) -- (1.2,0);
            
             \draw[-, very thick, dotted] (0,1) -- (1,0);

             \draw[-, nberblue, thick, dashed] (0, 0.464) -- (1-0.464, 0.464);
             \node at (-0.05, 0.464) {$\hat{p}$};

             \node at (0.16, 0.55) {taxation};
             \node at (0.18, 0.3) {delegation};

             \draw[-, ultra thick, nberblue] (0,0.744) -- (0, 1); 
             \draw[-, ultra thick, nberblue] (0,0.744) -- (1-0.464, 0.464);
             \draw[-, ultra thick, nberblue] (1,0) -- (1-0.464, 0.464);
             \node at (-0.1, 0.8) {{\color{nberblue}$P_\tau(q)$}};

             \node at (-0.05, 1) {$1$};
             \node at (1, -0.05) {$1$};
            \node at (-0.05, 1.2) {$p$};
            \node at (1.22, -0.05) {$q$};
            \node at (0.15, 0.95) {$P(q)$};

            \draw [->, thick,gray] (1.5,0) -- (1.5,1.2);
            \draw [->, thick,gray] (1.5,0) -- (2.7,0);

            \draw [-, gray, dashed] (1.5,0) -- (2.5,1);

            \node at (1.45, 1) {$1$};
            \node at (1.44, 0.52) {$0.5$};
            
            \node at (1.45, 1.2) {$p$};

             \draw [-, dashed, gray] (1.5,1) -- (2.5,1);

            \draw [-, nberblue!60, very thick, dashed] (1.5+0.192,0.464) -- (1.5+0.727,0.464+0.727-0.192);
            
            \fill[uncblue, opacity = 0.2] (1.5,0.464) -- (1.5 + 0.192, 0.464) -- (1.5+0.727,0.464+0.727-0.192) -- (1.5, 1);

            \node at (2.3, 1.05) {${\color{nberblue}P[q^*(c)]}$};
 
            \node at (1.5 + 0.192, -0.05) {$\hat{c}$};
            \node at (2.72, -0.05) {$c$};
            \node at (1.5+0.727, -0.05) {$\overline{c}$};

            \node at (2.3, 0.7) {$p^*(\cdot)$};
            \draw [-, dashed, gray] (1.5 +0.192, 0) -- (1.5 + 0.192, 0.464);
            \draw[gray, dashed] (1.5+0.727, 0) -- (1.5+0.727, 0.727);
            \draw[-, ultra thick] (1.5,0.464) -- (1.5 + 0.192, 0.464);
            \draw[-, ultra thick] (1.5+0.192,0.464) -- (1.5 + 0.727, 0.727);

            \draw [-, thick, purple, dotted] (1.5, 0.5) -- (2.5, 1);
            \node at (2.6, 0.95) {${\color{pennred}P[q_{LF}(\cdot)]}$};
           
            \draw [-, thick, purple!60, dotted] (1.5, 0.5) -- (2.5, 1);
             \node at (1.46, 0.45) {$\hat{p}$};
            
            \fill[pennred, opacity = 0.2] (1.5+0.727,0.464+0.727-0.192) --(1.5 + 0.192, 0.464) -- (1.5 + 0.727, 0.727); 
            \node at (0.6+1.5, 0.75) {\textbf{tax}}; 
        \end{tikzpicture}
    \end{center}
    \caption{\footnotesize Firm's inverse demand function $P_\tau(q)$ (left) and the induced pricing strategy $p^*(c)$ (right). $P(q) = 1 - q$, $c \sim U[0,1]$, $\alpha = 1$.}
    \label{fig:uniform}
\end{figure}

The left panel of Figure \ref{fig:uniform} shows how the optimal unit tax reshapes the firm’s inverse demand curve. The benchmark price $\hat{p}$ partitions the price space into a taxation region and a delegation region. In the taxation region, higher firm prices trigger higher unit taxes, generating lower demand than under the market demand curve. In the delegation region, any price posted by the firm yields market demand. Relative to the original inverse demand $P(q)$, the regulated inverse demand is more price sensitive just above $\hat{p}$, discouraging many medium-cost types from deviating upward.

The right panel shows the firm’s pricing strategy. The red dotted line is the laissez-faire price schedule $P[q_{LF}(c)]$. The black segments depict the policy-induced pricing function $p^*(c)$. In this example, costs are never low enough for the firm to choose its laissez-faire price $P(q_{LF}(c)) < \hat{p}$.\footnote{
Appendix \ref{apx:normal} provides an example in which some low-cost types do use laissez-faire prices.
}
The consumer price $P[q^*(c)] = p^*(c) + \tau(p^*(c))$ coincides with the firm price for $p^*(c) \le \hat{p}$ (i.e., $c \le \hat{c}$), and diverges in the taxation region (blue dashed line). The vertical gap between the consumer and firm prices (the red area bounded by $c=\overline{c}$, $p^*(c)$, and $P[q^*(c)]$) represents ex ante tax revenue.
The blue region (bounded by the $p$-axis, $p^*(c)$, and $P[q^*(c)]$) shows where trade occurs. Relative to laissez-faire pricing, the optimal policy lowers firm prices, but consumer prices fall only for medium-to-low-cost firm types, where the allocation moves closer to ex post efficiency ($P[q(c)] = c$, the 45° line). For high-cost firms, consumer prices exceed laissez-faire prices.

\vspace{.1in}\noindent{\it Comparison with Flexible Transfers and with No Transfer.} This example sharply contrasts with the outcomes under flexible transfers or under no transfers: even with $\alpha = 1$, the optimal tax-only policy induces firm prices that are strictly below the laissez-faire schedule, and the allocation is distorted relative to efficiency.

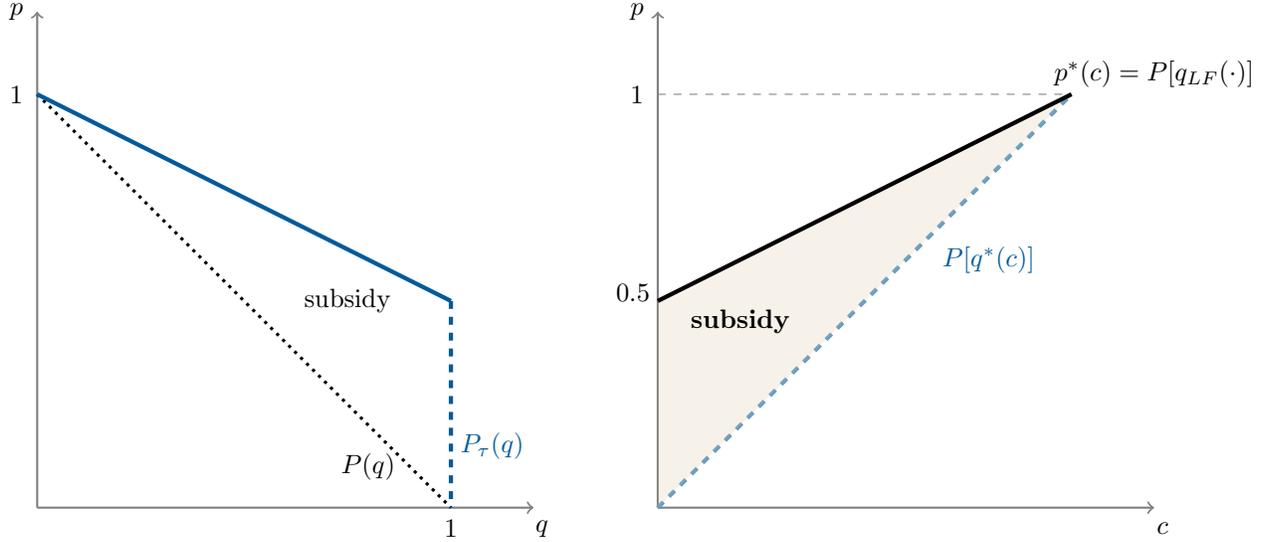
\begin{figure}[htbp]
    \begin{center}
               \begin{tikzpicture}[scale=5.5]
        \footnotesize
            \draw [->, thick,gray] (0,0) -- (0,1.2);
            \draw [->, thick,gray] (0,0) -- (1.2,0);
             \draw[-, very thick, dotted] (0,1) -- (1,0);
             \node at (0.75, 0.5) {subsidy};

             \draw[-, ultra thick, nberblue] (1,0.5) -- (0, 1); 
             \draw[-, ultra thick, dashed,nberblue] (1,0.5) -- (1, 0); 
             \node at (1.1, 0.15) {{\color{nberblue}$P_\tau(q)$}};
             
             \node at (-0.05, 1) {$1$};
             \node at (1, -0.05) {$1$};
            \node at (-0.05, 1.2) {$p$};
            \node at (1.22, -0.05) {$q$};
            \node at (0.8, 0.1) {$P(q)$};

            \draw [->, thick,gray] (1.5,0) -- (1.5,1.2);
            \draw [->, thick,gray] (1.5,0) -- (2.7,0);

            \draw [-, gray, dashed] (1.5,0) -- (2.5,1);

            \node at (1.45, 1) {$1$};
            \node at (1.44, 0.52) {$0.5$};
            
            \node at (1.45, 1.2) {$p$};
            
             \draw [-, dashed, gray] (1.5,1) -- (2.5,1);

            \draw [-, nberblue!60, ultra thick, dashed] (1.5, 0) -- (2.5, 1);
           
            \node at (2.3, 0.6) {${\color{nberblue}P[q^*(c)]}$};
        
            \node at (2.72, -0.05) {$c$};

            \fill[beige, opacity=0.2] (1.5, 0) -- (1.5, 0.5) -- (2.5, 1);

            \draw [-, thick, purple, dotted] (1.5, 0.5) -- (2.5, 1);
            \draw[-, ultra thick] (1.5, 0.5) -- (2.5, 1);
            \node at (2.7, 1.05) {$p^*(c) = P[q_{LF}(\cdot)]$};

            \node at (0.2+1.5, 0.45) {\textbf{subsidy}}; 

        \end{tikzpicture}
    \end{center}
    \caption{\footnotesize Firm's inverse demand function $P_\tau(q)$ (left) and the induced pricing strategy $p^*(c)$ (right) with flexible transfers according to \cite{BaronMyerson82}. $P(q) = 1 - q$, $c \sim U[0,1]$, $\alpha = 1$.}
    \label{fig:flexible}
\end{figure}

Under flexible transfers (Figure \ref{fig:flexible}), the regulator implements ex post efficient quantities using a menu of subsidies when $\alpha = 1$. The consumer price schedule satisfies $P[q(c)] = c$ for all $c$, maximizing gains from trade. To ensure truthful reporting, lower-cost types must receive larger subsidies; taxes are never used, as they would introduce demand distortions. The firm’s pricing strategy coincides with the laissez-faire one, $P[q_{LF}(c)]$.
The subsidy schedule (as a function of cost) equals $P[q_{LF}(c)] - c$. In other words, the regulator pays the information rent with fiscal budget.
The beige area in Figure \ref{fig:flexible} represents the welfare loss from the fiscal cost of subsidies.

Under no transfers, consumer and firm prices coincide. The regulator cannot improve upon the laissez-faire allocation in this example,\footnote{
See Amador and Bagwell (2022), footnote 33, p. 1744.
}
so $p^*(c) = P(q_{LF}(c)) = P(q^*(c))$.
Welfare loss arises solely from the monopolist’s quantity distortion (the beige area in Figure \ref{fig:flexible}).
In both the flexible-transfer and no-transfer settings, (i) taxes are never used, and (ii) firms always use laissez-faire pricing.

One might therefore expect that, without subsidies, the regulator would be powerless and that the optimal policy would revert to laissez-faire. Instead, taxation alone still delivers improvements. At the cost of excluding high-cost firms, the regulator depresses firm prices for lower-cost types and expands gains from trade. In Figure \ref{fig:uniform}, the laissez-faire price $P[q_{LF}(c)]$ crosses the optimal consumer price $P[q^*(c)]$ exactly once.
The region where the laissez-faire price lies above $P[q^*(c)]$ (part of the blue region) represents increased gains from trade; the region bounded by $P[q_{LF}(c)], P[q^*(c)]$, and $p=1$ represents the loss of gains from trade; and the red region denotes tax revenue. The optimal policy trades off these gains and losses.

Another perspective is that, when only taxes are available, the regulator behaves as though strongly redistributive.
Without subsidies, the regulator optimally transfers surplus from the firm to consumers—even when $\alpha = 1$. Such redistribution does not arise under no transfer or under flexible transfers. This reflects the substitutability between taxes and subsidies: with $\alpha = 1$, subsidies create incentives without distortion, whereas taxes create incentives only by distorting output. In this setting, taxation is the second-best instrument, and its benefits (sharper incentives for low prices) outweigh its costs (underproduction).

\bigskip

Finally, we discuss why the progressive price cap improves from a hard price cap. Both policies rely on exclusion, but a hard price cap excludes every firm with cost above the cap, and it does so merely to discipline marginally high-cost firms. Many low-cost types are unconstrained by a hard cap.
The tax-implemented progressive cap is instead a soft exclusion policy: it expands the range of feasible policies beyond the hard-cap case. Fewer high-cost types need to be excluded to achieve the same downward pressure on prices. Consequently, the optimal progressive cap uses a lower benchmark price, binds more low-cost types (improving affordability), and excludes fewer high-cost types (improving access and profitability).

\section{Conclusion}
Constraints on policy instruments have been recognized for decades as posing important practical and theoretical challenges, yet they have received relatively limited formal treatment. In their classic textbook (p.\ 155), \cite{LaffontTiroleBook93} note one such challenge:

\begin{quote}
[We have] examined pricing by a regulated natural monopoly… The analysis is a bit more complex, however, under a budget constraint because of the endogeneity of the shadow cost; much work remains to be done in this case.
\end{quote}

This paper contributes to this broader agenda by studying one such constraint within a workhorse model of monopoly regulation. In particular, imposing a no-subsidy requirement—one form of budget balance for the regulator—changes the structure of optimal regulation in systematic ways. We identify when intervention is desirable and how it should be designed. By characterizing the market environments in which regulation is beneficial or unnecessary, we show how taxes and subsidies can act as substitutes in some settings and as complements in others. Non-intervention can be optimal when demand is relatively price-insensitive, when high costs are sufficiently likely, or when the regulator assigns limited weight to redistribution. When intervention is warranted, we provide simple sufficient conditions under which the regulator implements a progressive price cap via a unit tax, balancing access, affordability, and profitability.

Future research could investigate how other policy constraints—such as bans on taxation, limits on subsidies, or caps on tax rates—affect optimal regulatory design. It could also explore these issues under alternative sources of uncertainty, such as hidden type-enhancing effort, demand uncertainty, or uncertainty in fixed costs.

\appendix

\section{Proof of Lemma \ref{lem:lf}}\label{appendix}

    Since $\tau(p) \equiv 0$, the demand function becomes $q(p) =P^{-1}(p)$. Since $q$ and $p$ has one-to-one relationship, the firm can maximize its profit over $q$ instead of $p$: \[
    \max_{q\in [0,\qmax]} (P(q) - c) q - k \cdot \mathbbm{1}_{q > 0},
    \]
    We can solve this problem in two steps. First, conditional on $q > 0$, the fixed cost $c$ is sunk. Then the problem becomes $\max_{q\in (0,\qmax]} (P(q) - c) q$. Ignore the primal constraint $q \in (0,\qmax]$ for now. By Assumption \ref{assumption_concave} (i)(ii), the objective is strictly concave defined on a compact feasible set. Hence, the sufficient and necessary first-order condition that characterizes solution $\hat{q}(c)$ to the problem without the primal constraint is \[
    P(q) = c  - q \cdot P'(q).
    \]
   By Assumption \ref{assumption_concave} (i), the solution $\hat{q}(c)$ is unique for each $c$. Standard implicit function theorem implies that $\hat{q}(c)$ strictly decreases in $c$. Furthermore, the envelope theorem implies that the optimal gross profit $\hat{\Pi}(c) = [P(\hat{q}(c)) - c] \hat{q}(c)$ is decreasing in $c$.

    Next, we consider whether the firm chooses $q > 0$ or $q = 0$. The firm with a cost $c$ will choose $q_{LF}(c) = \hat{q}(c)$ if $\hat{\Pi}(c) \geq k$, and will choose $q_{LF}(c) = 0$ otherwise.  We claim there is a unique cutoff $\overline{c}_{LF} \in (0,1]$ such that the firm chooses $q_{LF}(c) = \hat{q}(c)$ if and only if $c \leq \overline{c}_{LF}$. Since $\hat{\Pi}(c)$ is decreasing, it suffices to require $\hat{\Pi}(0) > k$, which is implied by Assumption \ref{assumption_concave} (iii). Moreover, the cutoff $\overline{c}_{LF}$ decreases as $k$ increases.

    Finally, we need to verify the solution candidate in the above analysis is feasible, i.e., $q_{LF}(c) \in [0, \qmax]$. It suffices to show that $\hat{q}(c) \leq \qmax$. Since $\hat{q}(c)$ is decreasing in $c$, we only need to show that $\hat{q}(0) \leq \qmax$. Recall that $\hat{q}(0)$ solves $P(q) + qP'(q) = 0$, the left-hand side of which decreases in $q$ because of concavity. Suppose that $\hat{q}(0) = \qmax$, the left-hand side becomes $\qmax P'(\qmax)< 0$, because, by Assumption \ref{assumption_concave} (ii), $P(\qmax) = 0$. We conclude that $\hat{q}(0) \leq \qmax$. To sum up, there is a cutoff $\overline{c}_{LF} \in (0,1]$ such that \[
   q_{LF}(c) = \begin{cases}
 \hat{q}(c),& c \leq \overline{c}_{LF}, \\
 0,& c > \overline{c}_{LF}
\end{cases}.
\]
Since $P(\cdot)$ is strictly decreasing, we conclude that the laissez-faire pricing strategy $P(q_{LF}(c))$ is strictly increasing in $c$.
This finishes the proof.

\section{Proof of Proposition \ref{prop_when} and \ref{prop_main}}\label{apx:main}
We first reformulate the mechanism design problem to employ the optimal control technique, then use the sufficient and necessary conditions to pin down the policies. Specifically, we prove Proposition \ref{prop_when} by proposing multipliers that are compatible with laissez-faire allocation, exploiting the sufficiency of the condition. We prove Proposition \ref{prop_main} by characterizing necessary properties of the optimal solution. Finally, we use two-step optimization to characterize the optimal policy.

Lemma \ref{lem:opt_ctrl_reform} claims that we can rewrite the regulator's design problem into an optimal control problem. In particular, the IC constraint \eqref{directIC} can be transformed into a ``law-of-motion'' condition wherever the demand policy $\demand(c)$ is continuous (i.e., condition \eqref{controlProfit}). We also allow for discontinuity in the demand policy, where the condition \eqref{controlProfit} will be invalid, and we handle the IC constraint using condition \eqref{IC_at_jump} instead. 
Notably, the optimal control problem is a regulation problem choosing $(q, \Pi)$ with firm's payoff is bounded above by the realized profit in the market, effectively ruling out lump-sum subsidies. If we add a ``no tax'' constraint\footnote{
The firm's profit cannot be lower than the profit associated with the market transaction, ruling out lump-sum excise taxes.
} \[
q(c)\cdot [P(q(c)) - c] - \fcost \cdot \mathbbm{1}_{q(c) > 0} \leq \Pi(c), \forall c\in[0,1],
\]
then the problem is equivalent to that of 
\cite{AmadorBagwell22Regulation}. Methodologically, we use Pontryagin approach to a mechanism problem with type-dependent constraints (e.g., \cite{XiaoPontryagin}; \cite{Xiao25JMP}).  

\begin{lem}\label{lem:opt_ctrl_reform}
The regulator's mechanism design problem \eqref{directOBJ} is equivalent to the following optimal control problem with state variables $(\Pi, \demand)$ and control variable $\control$, where $\Pi\colon [0,1] \rightarrow \R$ is continuous and piecewise differentiable, $\demand\colon [0,1] \rightarrow [0,\qmax]$ is piecewise differentiable, and $\control\colon [0,1] \rightarrow \R_-$ is piecewise continuous.
\begin{equation}
      \max_{\Pi, \demand, z}\int_0^1 \left[V(\demand(c)) - c\demand(c) -  \fcost \cdot \mathbbm{1}_{q(c) >0} - (1-\alpha)\Pi(c) \right]f(c)dc \label{controlOBJ}
  \end{equation}
  subject to, for all $c\in [0,1]$ where $\Pi'(c)$ and $q'(c)$ exist,
  \begin{align}
      &\Pi'(c) = -\demand(c),\; \Pi(0) \text{ free}, \Pi(1) = 0 \label{controlProfit}\\
      & \demand'(c) = \control(c) \leq 0,\; \demand(0) \text{ free}, \demand(1) \geq 0 \label{controlDemand}\\
      & q(c)\cdot [P(q(c)) - c] - \fcost \cdot \mathbbm{1}_{q(c)>0} \geq \Pi(c) \label{controlFS}
  \end{align}
  and for any $\sigma \in [0,1]$ where $\Pi'(\sigma)$ or $\demand'(\sigma)$ does not exist, we require \begin{equation}
        \demand(\sigma_-) \geq \demand(\sigma_+) \text{ and } \sigma = \arg\max_x [p(x) - \sigma]\demand(x), \label{IC_at_jump}
  \end{equation}
  where $\demand(\sigma_+) := \lim_{c \rightarrow \sigma_+} \demand(c)$, and $\demand(\sigma_-) := \lim_{c \rightarrow \sigma_-}\demand(c)$.
\end{lem}
\begin{proof}
    For notational convenience, define the firm's revenue as $\revenue(c) := p(c)\demand(c)$. Starting with the direct mechanism design \eqref{directOBJ}, its \eqref{directIC} implies that \[
\revenue(c_1) - c_1 \demand(c_1) \geq \revenue(c_2) - c_1 \demand(c_2), \; \revenue(c_2) - c_2 \demand(c_2) \geq \revenue(c_1) - c_2 \demand(c_1), \forall c_1, c_2 \in [0,1],
\]
where the fixed cost $\fcost$ cancels out.
Manipulating the above inequalities, we get $[\demand(c_2) - \demand(c_1)](c_2 - c_1) \leq 0$, implying that $\demand(\cdot)$ is nonincreasing, which in turn implies $\revenue(\cdot)$ is nonincreasing. The monotonicity of $\demand(\cdot)$ along with the verbatim \eqref{directIC} leads to condition \eqref{IC_at_jump}, without requiring differentiability of $\Pi(c)$ or continuity of $q(c)$. 

Since both $\demand(c)$ and $\revenue(c)$ are decreasing and bounded, they are piecewise differentiable. Recall that the value function for the firm is\[
\Pi(c) = \max_x (p(x) - c)\demand(x) - \fcost \cdot \mathbbm{1}_{q(x)>0}.
\] 
 Then $\Pi(c) = \max_x \revenue(x) - c\demand(x)- \fcost \cdot \mathbbm{1}_{q(x)>0}$.
The first-order necessary condition of the above problem requires $[\revenue'(x) - c\demand'(x)]|_{x=c} = 0$. We claim that the first-order condition along with nonincreasing condition on $\demand(c)$ is sufficient and necessary for the (global) \eqref{directIC} constraint for all active types (i.e., $q(c) > 0$). We have already shown the necessity. For sufficiency, suppose that for some $c \in [0,1]$ satisfying $\revenue'(c) - c\demand'(c) = 0$, there is some $y \neq c$ such that $y\in [0,1]$ and that \begin{equation}
    \revenue(c) - c \demand(c) < \revenue(y) - c \demand(y). \label{cond:deviation}
\end{equation}
We show that nondecreasing $\demand(\cdot)$ will lead to a contradiction following \cite{LaffontTiroleBook93} (p.157). Condition \eqref{cond:deviation} is equivalent to \[
\int_{c}^y \frac{d}{dx}[\revenue(x) - c \demand(x)]dx > 0.
\]
The first-order condition implies $\frac{d}{dx}[\revenue(x) - c \demand(x)]|_{x=c} = 0$, and thus \[
\int_{c}^y \left(\frac{d}{dx}[\revenue(x) - c \demand(x)] - \frac{d}{dx}[\revenue(x) - c \demand(x)]|_{c = x}  \right)dx > 0,
\]
which is equivalent to \[
\int_{c}^y \int_{x}^c \frac{d}{dt}[\revenue'(x) - t \demand'(x)]dt dx > 0,
\]
i.e., $\int_c^y\int_c^x(q'(x))dtdx > 0$, a contradiction to $q(\cdot)$ being nonincreasing. Thus, we conclude that the first-order condition $[\revenue'(x) - c\demand'(x)]|_{x=c} = 0$ and $\demand(c)$ being nonincreasing is sufficient for global \eqref{directIC} for all active types. The sufficiency and necessity is easy to extend to all types $c \in [0,1]$. Consider some generic $c_1, c_2 \in [0,1]$.
Since $q(c)$ is decreasing, $q(c_1) > 0$ and $q(c_2) = 0$ implies that $c_1 < c_2$. The no-subsidy constraint implies $\Pi(c_2) \leq 0$. Therefore, to make sure $c_1$ does not misreport its type to be $c_2$, we only need $\Pi(c_1) \geq 0$, implied by \eqref{directIR}.

Note that the first-order condition is equivalent to the envelope condition, implying the latter is also sufficient and necessary for optimality. We will use the envelope condition in our analysis for convenience. Recall that $\Pi(c) = \max_x \revenue(x) - c\demand(x) - \fcost\cdot \mathbbm{1}_{q(x) >0}$.
Since the maximand is affine (and thus convex) in $c$, $\Pi(c)$ is a convex function, which implies it is continuous and almost everywhere differentiable. 
By envelope theorem (e.g., \cite{MilgromSegal02Env}), $\Pi'(c) = -q(c)$ wherever differentiable. Therefore, a set of sufficient and necessary conditions for global \eqref{directIC} is that (i) $\Pi(c)$ is continuous and nonnegative; (ii) $\Pi'(c) = -q(c)$ wherever differentiable; and (iii) $q(c)$ is nonincreasing. We summarize these condition concisely in the form of ``law of motion'', \eqref{controlProfit}\eqref{controlDemand}, in the optimal control problem.

We have already shown in the text that constraint \eqref{directFS} is equivalent to $q(c) [P(q(c)) - c] - \fcost \cdot \mathbbm{1}_{q(c)>0} \geq \Pi(c)$, i.e., the unit from the market less the cost (left-hand side) must be no less than firm's profit (right-hand side), as a result of no subsidy. 
This constitutes to the current no-subsidy constraint \eqref{controlFS}.

Finally, $\Pi'(c) = -\demand(c) \leq 0$ and continuity implies $\Pi(c)$ is decreasing in $c$. Constraint \eqref{directIR} is thus equivalent to requiring $\Pi(1) \geq 0$.
Since $\alpha \leq 1$, it is always improves the welfare if the regulator transfers the profit of the highest-cost firm to the consumers, and because of quasi-linearity of firm's payoff, this does not change demand allocation. 
Therefore, it is without loss to set $\Pi(1) = 0$. 

On the other hand, there is no restriction of $q(1)$ other than being nonnegative. Thus we have the terminal conditions in \eqref{controlProfit} \eqref{controlDemand}.
\end{proof}

Lemma \ref{lem:IR_jump} simplifies our analysis by identifying a cutoff $\overline{c} \in (0,\overline{c}_{LF}]$ for exclusion, at which $q(\cdot)$ can jump from positive value to zero. 

\begin{lem}[\cite{AmadorBagwell22Regulation}] \label{lem:IR_jump}
    In any incentive compatible $q(\cdot)$, there exists a cutoff $\overline{c}   \in (0, \overline{c}_{LF}]$ such that $q(c) = 0$ for $c > \overline{c}$ and $q(c) > 0$ for $c < \overline{c}$. Moreover, if $\overline{c} \in (0, 1)$ and $\fcost>0$, then $\Pi(\overline{c}) = 0$ and $q(\overline{c}) > 0$.
\end{lem}
\begin{proof}
    Since $q(c)$ is monotonically decreasing and that $q(c)$ is bounded below at $0$, there is at most one cutoff $\overline{c}$ such that $q(c) = 0$ if and only if $c > \overline{c}$. We claim $\overline{c} > 0$; otherwise, all types have $q(c) = 0$ and this is clearly worse than the laissez-faire baseline, which is always feasible. Meanwhile, $\overline{c} \leq \overline{c}_{LF}$ because $\max_{q}q[P(q) - \overline{c}_{LF}]  = k$. The left-hand side strictly decreases in $c$ when $q > 0$. Hence, $\overline{c} > \overline{c}_{LF}$ violates the no-subsidy constraint. 
    If $\overline{c} \in (0, \overline{c}_{LF})$, then by definition $q(c) = 0$ for $c \in (\overline{c}, \overline{c}_{LF}]$, which implies $\Pi(c) = 0$. By continuity of $\Pi(c)$, we conclude that $\Pi(\overline{c}) = 0$. The no-subsidy constraint then implies $q(c)[P(q(c))- c] \geq  k$ on $(\overline{c} - \epsilon, \overline{c})$, for some small $\epsilon> 0$. When $k > 0$, it is necessary that $q(c)$ is bounded away from $0$ for $c \in (\overline{c} - \epsilon, \overline{c})$. Without loss, we can set $q(\overline{c}) = \lim_{c \rightarrow \overline{c}_-}q(c) > 0$. That is, $q(\cdot)$ jumps from positive value to zero at $\overline{c}$.
\end{proof}

Using Lemma \ref{lem:opt_ctrl_reform} and \ref{lem:IR_jump}, we can solve the optimal control problem in two steps: (i) Fixing $\overline{c} \in [0,\overline{c}_{LF}]$, solve a truncated version of optimal control \eqref{controlOBJ}-\eqref{IC_at_jump} by replacing the terminal type with $\overline{c}$ and the terminal values as $\Pi(\overline{c}) = 0, q(\overline{c}) \geq 0$; (ii) optimize over $\overline{c}$.

For the truncated version of control problem \eqref{controlOBJ}-\eqref{IC_at_jump}, we derive a set of necessary and sufficient conditions. Define the constraint value of \eqref{controlFS} to be \[
g(q, \Pi, c) = qP(q) - cq - \Pi - \fcost.
\]
By Assumption \ref{assumption_concave}, $g$ is concave in $(q, \Pi)$, and thus quasi-concave. On the other hand, the integrand of the objective is concave in $(q, \Pi)$  because it is linear in $\Pi$, and the second-order derivative with respect to $q$ is $P'(q) < 0$. As a result, the following Hamiltonian is concave in $(q, \Pi)$:  \[
H = [V(\demand) - c\demand - (1-\alpha)\Pi - \fcost]f - \lambda_\Pi \demand + \lambda_{\demand} \control. 
\]
Define the Lagrangian to be \[
L = H + \gamma(c) \cdot g(\demand, \Pi, c),
\]
where $\gamma(c) \geq 0$ is some piecewise continuous function of $c$.
By Theorem 5.1 in \cite{SeierstadSydsaeter93OptCtrl} (pp.\ 317-319), we derive the following sufficient and necessary conditions\footnote{
The necessity is guaranteed according to  pp.\ 335 - 336 of \cite{SeierstadSydsaeter93OptCtrl}.
The variable $\beta$ in \eqref{ctrl:lambda_pi} and \eqref{ctrl:lambda_q} is a scale parameter for the jump of costate at $c=0$.
By conditions (5.37) - (5.38) of \cite{SeierstadSydsaeter93OptCtrl}, the costate cannot jump at $c \in (0,1]$. Moreover, 
the no-subsidy constraint \eqref{controlFS} being slack at $c = 0$ or $q'(0)$  not existing implies $\beta = 0$. 
In condition \eqref{ctrl:lambda_pi}, $\lambda_\Pi(0_+) = -\beta \frac{\partial g(q, \Pi, 0)}{\partial \Pi} = \beta \geq 0$.
For condition \eqref{ctrl:lambda_q} to be consistent with the primal condition \eqref{ctrl:control}, we need $\lambda_\demand(0_+) = -\beta\frac{\partial g(q, \Pi, c)}{\partial q}|_{c=0} \geq 0$, which imposes a requirement on $q$, i.e., $q(0) \geq q_{LF}(0)$.
} from the Lagrangian: $\forall c \in [0,1]$ \begin{align}
    & \lambda_\Pi' = (1-\alpha)f +\gamma, \text{ with }\lambda_\Pi(0) = 0, \lambda_\Pi(0_+) = \beta \geq 0, \lambda_\Pi(1) \text{ free }  \label{ctrl:lambda_pi}\\
      & \lambda_\demand' = -(P(\demand) - c)f +\lambda_\Pi - \gamma \left[P(q) + qP'(q) - c\right], \nonumber \\
      & \qquad \qquad \qquad  \text{ with }\lambda_\demand(0) = 0, \lambda_\demand(0_+) = -\beta \frac{\partial g(\demand, \Pi, c)}{\partial \demand}, \lambda_\demand(1)\geq 0  \label{ctrl:lambda_q}\\
   &\lambda_\demand \geq 0, z\lambda_\demand = 0 \text{ with } z(c) < 0 \text{ or } q(c_-) > q(c_+) \Rightarrow \lambda_\demand(c) = 0\label{ctrl:control}\\
   &\gamma \geq 0, \gamma \left[\demand P(\demand) - c\demand - \Pi - \fcost\right]= 0 \text{ for } c< \overline{c}\label{ctrl:fs}
\end{align}
Note that we allow for discontinuity of $q(c)$ and the corresponding condition is in \eqref{ctrl:control} (see \cite{XiaoPontryagin} (p.\ 7)). 
We are ready to prove Proposition \ref{prop_when}, which provides a simple condition for the scope of intervention. 

\begin{proof}[Proof of Proposition \ref{prop_when}]
    Recall that under the laissez-faire baseline (Lemma \ref{lem:lf}), $q(c) = q_{LF}(c)$ and $\Pi(c) = q(c)[P(q(c)) - c] - \fcostapp$. We propose the multiplier function $\gamma(c)$ that is consistent of the laissez-faire solution, and impose conditions for the multiplier function to be valid. 
      
    By the first-order condition of the laissez-faire problem, we have \[
P(q_{LF}(c)) = c - q_{LF}(c) \cdot P'(q_{LF}(c)), \forall c \in [0, \overline{c}_{LF}],
\]
which implies that the term $P(q) + qP'(q) - c$ vanishes in \eqref{ctrl:lambda_q}. Plugging the solution candidate into condition \eqref{ctrl:lambda_q}, we get
 \[
\lambda_q'(c) = -[P(q_{LF}(c)) - c]f(c) + \lambda_\Pi(c),
\]
with \[
\lambda_\Pi(c) = (1-\alpha)F(c) + \Gamma(c) + \beta,
\]
where the cumulative multiplier function $\Gamma(c) := \int_0^c \gamma(x)dx \geq 0$ is nondecreasing in $c$.
Since the laissez-faire demand is strictly decreasing, we have $\lambda_q(c) \equiv 0$ for all $c \in [0,\overline{c}_{LF}]$, which in turn requires $\lambda_q'(c) \equiv 0$ for all $c \in [0,\overline{c}_{LF}]$, i.e.,\begin{equation}
    \Gamma(c) + \beta = [P(q_{LF}(c)) - c]f(c) - (1-\alpha)F(c), \forall c \in [0,\overline{c}_{LF}], \label{lf_optimal_derivative}
\end{equation}
which proposes a guess of the cumulative multiplier. 
A sufficient and necessary condition for the laissez-faire allocation to be truncated optimal (for $\overline{c} = \overline{c}_{LF}$) is to find a piecewise continuous, nonnegative function $\gamma(c)$ and some nonnegative constant $\beta$ that correspond to the above proposed cumulative multiplier $\Gamma(c) + \beta$, which is to require
\[
[P(q_{LF}(c)) - c]f(c) - (1-\alpha)F(c) \text{ is nondecreasing in } c \in [0, \overline{c}_{LF}],
\]
and is nonnegative when evaluated at $c=0$. Note that $P(q_{LF}(0) - 0)f(0) - (1-\alpha)F(0)$ is nonnegative because $F(0) = 0$ and $P(q_{LF}(c)) \geq c$. This results in condition \eqref{lf_optimal} in Proposition \ref{prop_when}. Finally, note that \eqref{lf_optimal} implies the same expression is nondecreasing over $[0, \overline{c}]$, for all $\overline{c} \leq \overline{c}_{LF}$. That is, the laissez-faire policy is truncated optimal for any truncation. Nevertheless, the truncation $\overline{c} = \overline{c}_{LF}$ is uniquely optimal because any $\overline{c} < \overline{c}_{LF}$ leads to more exclusion, which lowers the welfare. This finishes the proof.
\end{proof}

\medskip

Next, we turn to the optimal policy when condition \eqref{lf_optimal} fails. We pin down the behavior of the optimal policy by exploiting the necessity of  \eqref{ctrl:lambda_pi} - \eqref{ctrl:fs}. Recall that Lemma \ref{lem:IR_jump} shows that $q(\cdot)$ can jump at $\overline{c}$. The following Lemma \ref{lem:continuous} shows that $q(\cdot)$ does not have other discontinuity than $\overline{c}$. This is consistent with the intuition of continuous policies in \cite{BaronMyerson82}. A jump at any non-excluded type is either not incentive compatible, or not optimally extracting the firm's rent. Strict concavity of the market revenue function (Assumption \ref{assumption_concave}) and continuous density function $f$ of types constitute sufficient conditions to avoid interior jumps for our policy.

\begin{lem}\label{lem:continuous}
In Problem \eqref{controlOBJ} - \eqref{IC_at_jump}, it is without loss to restrict attention to a demand policy $q(c)$ that is continuous on $[0, \overline{c})$.
\end{lem}
\begin{proof}
Without loss, we can restrict attention to left- and right- continuous $\demand(c)$, because single-point jumps account to zero measure. 
    We first show that $\demand(c)$ has no jump for interior $c \in (0,\overline{c})$. Our argument follows \cite{MorrisShadmehr23Inspiring}, which dates back to \cite{Arrow666}. Suppose there is a jump of $\demand(c)$ at some $\sigma \in (0,\overline{c})$. Since $\demand(c)$ is decreasing, we must have $\demand(\sigma_+) < \demand(\sigma_-)$, where $q(\sigma_+) := \lim_{c \rightarrow \sigma_+} \demand(c)$, and $\demand(\sigma_-) := \lim_{c \rightarrow \sigma_-}\demand(c)$. 
    
    Now, we discuss with respect to whether the no-subsidy constraint \eqref{controlFS} binds around $\sigma$. 
    
    \vspace{.1in} \noindent{\it Case 1: The no-subsidy constraint  
    \eqref{controlFS} is slack for both $\demand(\sigma_-)$ and $\demand(\sigma_+)$.} By the complementary slackness condition \eqref{ctrl:fs},  $\gamma(\sigma_+) = \gamma(\sigma_-) = 0$. Recall that the ``law-of-motion'' \eqref{ctrl:lambda_q} of the costate $\lambda_q$ is given by
    \[
    \lambda_\demand'(c) = -(P(\demand(c)) - c)f(c) + \lambda_\Pi(c) - \gamma(c) \left[q(c)\cdot P'(q(c)) + P(q(c)) - c\right],
    \]
    for $c\in (\sigma-\epsilon, \sigma) \cup (\sigma, \sigma+\epsilon)$, for some small $\epsilon>0$; and that \[
    \lambda_\Pi(c) = \underbrace{\lambda_\Pi(0_+)}_{=\beta} + \int_0^c \lambda_\Pi'(x)dx =  (1-\alpha)F(c) + \underbrace{ \int_0^c\gamma(x)dx}_{:= \Gamma(c)} + \beta.
    \]
  The above expression implies $\lambda_\Pi(\sigma_-) = \lambda_\Pi(\sigma_+)$. Therefore, when $\gamma(\sigma_+) = \gamma(\sigma_-) = 0$, the difference between $\lambda_\demand'(\sigma_-)$ and $\lambda_\demand'(\sigma_+)$ is governed 
  by the term $-(P(\demand(c)) - c)f(c)$. Because (i) $\demand(\sigma_-) > \demand(\sigma_+)$, (ii) $f$ is continuous at $\sigma$, and (iii) $-P(\cdot)$ is strictly increasing, we conclude that 
    \begin{equation}
         \lambda_\demand'(\sigma_-)  > \lambda_\demand'(\sigma_+). \label{continuity_contradiction}
    \end{equation}
    On the other hand, $\lambda_\demand(\sigma) = 0$ at the jump point of its corresponding state $\demand(\cdot)$.\footnote{
See condition (3.74) of Theorem 3.7 in Chapter 3,  \cite{SeierstadSydsaeter93OptCtrl} (p.\ 196-198). Also see condition (4.13) in  \cite{Hellwig10ECMA} and \cite{XiaoPontryagin} (p.\ 7).
    }
    Since $\lambda_\demand(\sigma) \geq 0$ by the primal condition \eqref{ctrl:control}, we have\[
     \lambda_\demand'(\sigma_-)\leq 0, \lambda_\demand'(\sigma_+)\geq 0,
    \] 
    A contradiction to condition \eqref{continuity_contradiction}. Therefore, there cannot be a jump of $\demand(\cdot)$ at any $\sigma \in (0,\overline{c})$ when the no-subsidy constraint \eqref{controlFS} is slack for both $\demand(\sigma_-)$ and $\demand(\sigma_+)$.

    \smallskip
    
    In each of the remaining cases, we establish condition \eqref{continuity_contradiction} and thus reach a contradiction as in Case 1.

    \vspace{.1in} \noindent{\it Case 2: The no-subsidy constraint  
    \eqref{controlFS} binds for both $q(\sigma_-)$ and $q(\sigma_+)$.} We first show that the only incentive compatible situation is $q(\sigma_-) > q_{LF}(\sigma) > q(\sigma_+)$ for this case, i.e., the jump should ``cross'' the laissez-faire demand from above.
    Note that the firm must be indifferent at the jump point:
    \[
    q(\sigma_-)(P(q(\sigma_-)) - \sigma) - \fcost= q(\sigma_+)(P(q(\sigma_+)) - \sigma) - \fcost,
    \]
    where $\fcost$ appears on both sides of equation whenever we compare profits and is henceforth omitted.
    Suppose, for the sake of contradiction, that $q_{LF}(\sigma) \geq q(\sigma_-)  > q(\sigma_+)$. 
    By strict concavity of the profit function, \[
    \frac{\partial[q(P(q) - \sigma)]}{\partial q}\bigg|_{q < q_{LF}(\sigma)} > 0,
    \] 
    which implies \[
      q(\sigma_-)(P(q(\sigma_-)) - \sigma) > q(\sigma_+)(P(q(\sigma_+)) - \sigma),
    \]
    because $q(\sigma_-) > q(\sigma_+)$. A contradiction. The infeasibility of $q(\sigma_-)  > q(\sigma_+) \geq q_{LF}(\sigma)$ follows a mirror argument. 
    
    Now we focus on $q(\sigma_-) > q_{LF}(\sigma) > q(\sigma_+)$.
    By the complementary slackness condition \eqref{ctrl:fs}, $\gamma(\sigma_-) \geq 0, \gamma(\sigma_+) \geq 0$. In Case 1, we have already shown that the first term $-(P(\demand(c)) - c)f(c)$ jumps down, and that $\lambda_\Pi(c)$ is continuous. What remains is to show \[
    -\gamma(\sigma_-) \left[q(\sigma_-)\cdot P'(q(\sigma_-)) + P(q(\sigma_-)) - \sigma\right] \geq -\gamma(\sigma_+) \left[q(\sigma_+)\cdot P'(q(\sigma_+)) + P(q(\sigma_+)) - \sigma\right],
    \]
    or equivalently,
    \[
    -\gamma(\sigma_-) ([qP(q)]'|_{q = q(\sigma_-)} - \sigma) \geq -\gamma(\sigma_+)([qP(q)]'|_{q = q(\sigma_+)} - \sigma).
    \]
 By strict concavity of $qP(q)$, we obtain $[qP(q)]'|_{q = q(\sigma_+)} > \sigma > [qP(q)]'|_{q = q(\sigma_-)}$, which implies the above inequality. Hence, we establish condition \eqref{continuity_contradiction}, leading to a contradiction. As a result, there cannot be a jump of $q(\cdot)$ at any $\sigma \in (0,\overline{c})$ when the no-subsidy constraint \eqref{controlFS} binds for both $q(\sigma_-)$ and $q(\sigma_+)$.

 \vspace{.1in} \noindent{\it Case 3: The no-subsidy constraint  
    \eqref{controlFS} is slack for $q(\sigma_-)$ and binds for $q(\sigma_+)$.} 
    Similar to the previous case,  it is not incentive compatible for $q(\sigma_-) > q(\sigma_+) \geq q_{LF}(\sigma)$. By strict concavity of the laissez-faire profit function, the profit is decreasing in $q$ when $q$ is larger than $q_{LF}$, i.e., \[
    q(\sigma_-)(P(q(\sigma_-)) - \sigma) < q(\sigma_+)(P(q(\sigma_+)) - \sigma).
    \]
    Note that $\Pi(\sigma_-) < q(\sigma_-)(P(q(\sigma_-)) - \sigma)$ because the no-subsidy constraint is slack, implying $\Pi(\sigma_-) < \Pi(\sigma_+)$, which violates the IC. Therefore, we know that at least $q(\sigma_+) < q_{LF}(c)$, regardless the relationship between $q(\sigma_-)$ and $q_{LF}(c)$. But then, 
    \[
    -\underbrace{\gamma(\sigma_-)}_{=0} \left[q(\sigma_-)\cdot P'(q(\sigma_-)) + P(q(\sigma_-)) - \sigma\right]  \geq -\gamma(\sigma_+) \left[q(\sigma_+)\cdot P'(q(\sigma_+)) + P(q(\sigma_+)) - \sigma\right], 
    \]
    the right-hand side of which is non-positive due to the same reasoning as Case 2. Again, we have condition \eqref{continuity_contradiction}, leading to a contradiction. Hence, there cannot be a jump of $q(\cdot)$ at any $\sigma \in (0,\overline{c})$ when the no-subsidy constraint \eqref{controlFS} is slack for $q(\sigma_-)$ and binds for $q(\sigma_+)$.

    \vspace{.1in} \noindent{\it Case 4: The no-subsidy constraint  
    \eqref{controlFS} binds for $\demand(\sigma_-)$ and is slack for $\demand(\sigma_+)$.} We first show that $q(\sigma_-) > q_{LF}(\sigma)$. Suppose not, then $q_{LF}(\sigma) \geq q(\sigma_-) > q(\sigma_+)$. By concavity of the laissez-faire profit function, we have \[
    q(\sigma_-)(P(q(\sigma_-)) - \sigma) > q(\sigma_+)(P(q(\sigma_+)) - \sigma),
    \]
    and the right-hand side is strictly greater than $\Pi(\sigma_+)$ because the no-subsidy constraint is slack, violating the IC constraint. We conclude that $q(\sigma_-) > q_{LF}(\sigma)$. Again, this implies 
\[
-\gamma(\sigma_-)\left[q(\sigma_-)\cdot P'(q(\sigma_-)) + P(q(\sigma_-)) - \sigma\right]  \geq -\underbrace{\gamma(\sigma_+)}_{=0} \left[q(\sigma_+)\cdot P'(q(\sigma_+)) + P(q(\sigma_+)) - \sigma\right], 
\]
the left-hand side of which is nonnegative following the same reasoning as Case 2. We establish condition \eqref{continuity_contradiction}, leading to a contradiction. Hence, there cannot be a jump of $q(\cdot)$ at any $\sigma \in (0,\overline{c})$ when the no-subsidy constraint \eqref{controlFS} binds for $q(\sigma_-)$ and is slack for  $q(\sigma_+)$.

   \vspace{.1in}\noindent To sum up, there cannot be a jump at any interior $c \in (0,\overline{c})$. Finally, it is without loss to set $\demand(0) = \demand(0_+)$ and $\demand(\overline{c}) = \demand(\overline{c}_-)$, for $c = 0$ or $\overline{c}$ has zero measure and does not affect the regulator's objective value. In conclusion, it is without loss to restrict attention to continuous demand policy $\demand(c)$ over $[0, \overline{c})$ for Problem \eqref{controlOBJ}-\eqref{IC_at_jump}.
\end{proof}

With continuity, we can show that the transformed the truncated version of problem \eqref{controlOBJ} - \eqref{IC_at_jump} admits an optimal solution. As a result, the negation of condition \eqref{lf_optimal} is sufficient and necessary for intervention to be optimal. 

\begin{lem}\label{lem:existence}
There exists an (essentially) unique optimal solution $(q^*, \Pi^*, z^*)$ for the truncated version of control problem \eqref{controlOBJ} - \eqref{IC_at_jump}. Furthermore, it is without loss to focus on deterministic mechanism.
\end{lem}
\begin{proof}
We first show existence by invoking Theorem 5.5 of \cite{SeierstadSydsaeter93OptCtrl}, and check the problem \eqref{controlOBJ} - \eqref{IC_at_jump} satisfies their conditions. Denote the pointwise welfare as \[
    W(q, \Pi, c) := V(q) - cq - \fcost - (1-\alpha)\Pi
    \]
    By Lemma \ref{lem:continuous}, the objective integrand $W(q, \Pi, c)f(c)$, the law-of-motion $-q$, and the constraint value $q[P(q) - c] - \fcost - \Pi$ are continuous in $(q, \Pi, z, c)$.  For technical purpose, we can assume that $z$ has a negative lower bound $-Z < 0$. Since $q$ is continuous, monotone, and bounded, we can find a large enough $Z$ so that the restricted problem with $z \in [-Z, 0]$ is equivalent to the original problem with $z \leq 0$. Thus, we establish a compact control set $[-Z,0]$. Also, $\Pi$ is bounded because $\Pi \leq qP(q) - cq -\fcost\leq qP(q)$.
    Finally, for each fixed $(q, \Pi, c)$, the set \[
    N(q, \Pi, [-Z, 0], c) := \{(W(q, \Pi, c)f(c) + \eta, -q, z)\colon \eta \leq 0, z \in [-Z, 0]\} \subseteq \R^3
    \]
    is a rectangle unbounded on one side, which is convex. Therefore, all conditions of Theorem 5.5 of \cite{SeierstadSydsaeter93OptCtrl} hold. As a result, there exists an optimal solution $(q^*, \Pi^*, z^*)$ for the optimal control problem \eqref{controlOBJ} - \eqref{IC_at_jump}.

    For uniqueness, we follow the approach of \cite{KangWatt24}.\footnote{
    We also note that unlike \cite{KangWatt24}, our regulation problem is a convex program with \textit{non}-majorization constraints: \[
q(c)P(q(c)) - cq(c) - \fcost \geq \int_c^1 q(x)dx.
\]
Therefore, we cannot proceed using the novel approach provided in \cite{KangWatt24}, neither can we adopt the approach of \cite{KleinerMajorizationECMA}.
    }
    Suppose there are two distinct optimal mechanisms $(q_1, \Pi_1)$ and $(q_2, \Pi_2)$ that differ for some positive-measure range of costs. Then each of the mechanism must satisfy IC and the no-subsidy constraint $qP(q) - cq \geq \Pi$. Consider an alternative mechanism $(\frac{q_1+q_2}{2}, \frac{\Pi_1+\Pi_2}{2})$. The new mechanism satisfies IC because of linearity. It also satisfies the no-subsidy constraint because $qP(q)$ is strictly concave, we have for each $c$ \[
    \frac{q_1+q_2}{2}P\left(\frac{q_1+q_2}{2}\right) - c\frac{q_1+q_2}{2} - \fcost > \frac{q_1P(q_1)+q_2P(q_2)}{2} - c\frac{q_1+q_2}{2} - \fcost \geq \frac{\Pi_1+\Pi_2}{2}.
    \]
    Since the objective is strictly concave, the mechanism  $(\frac{q_1+q_2}{2}, \frac{\Pi_1+\Pi_2}{2})$ strictly dominates both $(q_1, \Pi_1)$ and $(q_2, \Pi_2)$, contradicting to their optimality. This also implies the optimal deterministic mechanism dominates any stochastic mechanism. For any stochastic mechanism with support of a set of deterministic mechanisms, the arithmetic average of the deterministic mechanisms dominates the stochastic one according to the same argument. 

    Hence, we conclude that there is an essentially unique optimal solution $(q^*, \Pi^*, z^*)$ that is deterministic. 
\end{proof} 

Next, we characterize the behavior of the no-subsidy constraint. 
The following two lemmata show that the no-subsidy constraint \eqref{controlFS} binds for $c \rightarrow 0$ and it turns from binding to slack when and only when $q(c) > q_{LF}(c)$.

\begin{lem}\label{lem:no_tax_all}
Suppose condition \eqref{lf_optimal} in Proposition \ref{prop_when} is violated, then in any optimal solution $(q, \Pi)$, 
the no-subsidy constraint \eqref{controlFS} is slack for a positive-measure set of types. Nevertheless,  \eqref{controlFS} must bind for $[0, \epsilon]$ for some small $\epsilon> 0$.  
\end{lem}
\begin{proof}
    The first part of the claim is a direct consequence combining Proposition \ref{prop_when} with Lemma \ref{lem:existence}: Since condition \eqref{lf_optimal} is violated, the laissez-faire baseline is sub-optimal, which implies any truncated laissez-faire allocation over $[0, \overline{c}]$ with $\overline{c} < \overline{c}_{LF}$ is also suboptimal. Therefore, any optimal allocation must involve a positive-measure set of types for which the no-subsidy constraint is slack (i.e., intervention). 
    For the second part, suppose that on the contrary, \eqref{controlFS} is slack for $[0, \epsilon]$. By condition \eqref{ctrl:lambda_pi} and \eqref{ctrl:lambda_q}, 
    we have $P(q(c)) = c + (1-\alpha)\frac{F(c)}{f(c)}$ for $c \in [0, \epsilon]$. 
    On the other hand, \eqref{controlFS} being slack implies $q(c)P(q(c)) - cq(c) - \fcost > \Pi(c) = \int_c^1q(x)dx > 0$. However, this inequality is violated for $c \rightarrow 0$ because $\lim_{c\rightarrow0} [q(c)P(q(c)) - cq(c) - \fcost] = \lim_{c\rightarrow0} [q(c)F(c)/f(c) - \fcost] \leq 0$. A contradiction. Therefore, \eqref{controlFS} must bind for $[0, \epsilon]$ with some small $\epsilon > 0$.
\end{proof}

\begin{lem} \label{lem:no-subsidy-behavior}
Fix any optimal policy $(q, \Pi)$. There exists countably many cutoffs $0 = c_0 \leq c_1\leq c_2 \leq \cdots  \leq c_{2K} =  \overline{c}$, such that the no-subsidy constraint binds if and only if $c \in [c_{2k}, c_{2k+1}]$ for $k = 0, 1, ..., K$. Moreover, $q(c_{2k+1}) > q_{LF}(c_{2k+1})$ and $q(c_{2k+2}) < q_{LF}(c_{2k+2})$, for $k = 0, 1, ..., K$.
\end{lem}

\begin{proof}
 By Lemma \ref{lem:continuous}, the constraint value \[
g(\demand(c), \Pi(c), c) = \demand(c)P(\demand(c)) - c\demand(c) - \fcost - \Pi(c)
\]
is continuous on $[0, \overline{c})$ for any $\overline{c}\in (0, \overline{c}_{LF}]$. Taking total derivative of $g(q(c),\Pi(c),c)$ with respect to $c$, we obtain \begin{align*}
    \frac{d g(\demand(c), \Pi(c), c)}{d c}
    &=  \frac{\partial  g}{\partial \demand}\cdot \demand'(c)+ \frac{\partial  g}{\partial \Pi}\cdot \Pi'(c)  + \frac{\partial  g}{\partial c}\\
    &= [q(c)P'(q(c)) + P(q(c)) - c]q'(c) - 1\cdot [-q(c)] - q(c)\\
    &= [q(c)P'(q(c)) + P(q(c)) - c]q'(c),
\end{align*}
wherever it exists. For notation simplicity, we henceforth use $g(c) := g(q(c), \Pi(c),c)$. By Lemma \ref{lem:no_tax_all}, the no-subsidy constraint binds at $0$. Then, it suffices to characterize the transition points at which the constraint becomes slack or binding. 
To begin with, if there is some $\hat{c} \in (0,\overline{c})$ such that $g(\hat{c}) = 0$ and $g(\hat{c}+\epsilon) > 0$ for any small $\epsilon> 0$, then $q(\hat{c}) > q_{LF}(\hat{c})$.
To see this, note that $g(c)$ is piecewise differentiable, and thus we have $g'(\hat{c}_+) > 0$. 
Since $q'(c) \leq 0$ for all $c$, it must be that $q(\hat{c})P'(q(\hat{c}_+)) + P(q(\hat{c})) - \hat{c} < 0$. Strict concavity of $qP(q)$ implies $q(\hat{c}) > q_{LF}(\hat{c})$. Similarly, if there is some $\hat{c} \in (0,\overline{c})$ such that $g(\hat{c}) = 0$ and $g(\hat{c}- \epsilon) > 0$, then $q(\hat{c}) < q_{LF}(\hat{c})$. This finishes the proof.
\end{proof}

\paragraph{Single Crossing} Lemma \ref{lem:sc} shows that, with the sufficient condition that the \textit{cost density is nonincreasing}, the no-subsidy constraint has a single-crossing property. This is the minimal sufficient condition we identify that makes the no-subsidy constraint single-crossing.

We are going to use the following fact in the proof of Lemma \ref{lem:sc}: Observe that for the firm to be incentive compatible in the interval with a binding no-subsidy constraint, it is necessary to require either $q(c) = q_{LF}(c)$ or $q'(c) = 0$. To see this, substituting $\Pi(c)$ by binding no-subsidy constraint \eqref{controlFS} into IC, we have $[q(c)P(q(c)) - cq(c)]' = -q(c)$, which is satisfied either $q'(c) = 0$ (flat) or $P(q(c)) = c - q(c)P'(q(c))$ (laissez-faire).

\begin{lem}\label{lem:sc}
Suppose condition \eqref{lf_optimal} in Proposition \ref{prop_when} is violated. Assume that the density $f$ is nonincreasing.  
Fix any optimal policy $(q, \Pi)$. For the range $[0, \overline{c})$, there is exactly one cutoff $\hat{c} \in (0,\overline{c})$ such that the no-subsidy constraint \eqref{controlFS} binds if and only if $c \leq \hat{c}$.  
\end{lem}

\begin{proof}
By Lemma \ref{lem:no_tax_all} and continuity, there exists at least one cutoff $\hat{c} \in (0, \overline{c})$ for the no-subsidy constraint to shift from binding to slack, i.e., $g(\hat{c}) = 0$ and $g(c) > 0$ for $c \in (\hat{c}, \hat{c} + \epsilon)$ for small $\epsilon > 0$. We show that there is exactly one such cutoff, i.e.,  $g(c) > 0$ for all $c \in (\hat{c}, \overline{c})$. 

Suppose there is some proper subset of $[\hat{c}, \overline{c}]$ in which $g(c) = 0$, there must be two cutoffs $c_1 < c_2$, satisfying $g(c_1) = 0, g(c) > 0$ for $c \in (c_1 - \epsilon, c_1)$;  and $g(c_2) = 0, g(c) > 0$ for $c \in (c_2, c_2 + \epsilon)$. By Lemma \ref{lem:no-subsidy-behavior}, $q(c_1) < q_{LF}(c_1)$ and $q(c_2) > q_{LF}(c_2)$. According to the observation before this lemma, for the firm to be incentive compatible, we require $q'(c) = 0$ for $q(c) \neq q_{LF}(c)$ on $[c_1, c_2]$. Continuity (Lemma \ref{lem:continuous}) implies two cases: either $P(q(c_1)) = P(q(c_2))$, or there exists some $c_{LF}$ such that $P(q(c_1)) = P(q(c_{LF}))$ and $q(c) = q_{LF}(c)$ for $c \in [c_{LF}, c_{LF}+\epsilon)$.

Suppose that $P(q(c_1)) = P(q(c_2))$.
By condition \eqref{ctrl:lambda_q} and \eqref{ctrl:control}, $\lambda_q'(c) = 0$ requires \[
P(q(c)) = c +  (1-\alpha)\frac{F(c)}{f(c)} + \frac{\Gamma(c) + \beta}{f(c)},
\]
satisfied by both $c=c_1$ and $c=c_2$. By our supposition, $P(q(c))$ strictly increases in $(c_1 - \epsilon,c_1)$ and $(c_2, c_2 + \epsilon)$.
Let $\phi(c;\Gamma(x)) = c +  (1-\alpha)\frac{F(c)}{f(c)} + \frac{\Gamma(x) + \beta}{f(c)}$, for some constant $x \leq c$.
Since $f$ is nonincreasing, $\phi(c;\Gamma(x))$ is strictly increasing in $c$. We have
\[
P(q(c_2)) = \phi(c_2;\Gamma(c_2)) \geq \phi(c_2;\Gamma(c_1)) > \phi(c_1;\Gamma(c_1)) = P(q(c_1)),
\]
a contradiction to the supposition that $P(q(c_2)) = P(q(c_1))$.

Suppose that $P(q(c_1)) = P(q(c_{LF}))$ and $q(c) = q_{LF}(c)$ for $c \in [c_{LF}, c_{LF}+\epsilon)$. 
We know that $P(q(c_{LF})) = \phi(c_{LF}; \Gamma(c_{LF})) \geq \phi(c_{LF}; \Gamma(c_1)) > \phi(c_1; \Gamma(c_1)) = P(q(c_{1}))$. A contradiction. 

We thus establish that there is no proper subset of $[\hat{c}, \overline{c}]$ in which $g(c) = 0$. In other words, there is exactly one cutoff $\hat{c} \in (0, \overline{c})$ such that the no-subsidy constraint \eqref{controlFS} binds if and only if $c \leq \hat{c}$.
 \end{proof}

 \paragraph{Proof of Proposition \ref{prop_main}.} So far, we know if the cost density is nonincreasing, there is a unique cutoff $\hat{c} \in (0, \overline{c})$ for the no-subsidy constraint \eqref{controlFS} to bite. We first show that $\tau(p)$ increasing in the firm price $p$ is equivalent to the \textit{direct unit tax} $\tau(p(c)) := g(c)/q(c)$ increasing in $c$ over $[\hat{c}, \overline{c})$, which boils down to verifying $p(c)$ is increasing.
 By the definition of firm profit, $p(c) = c +  \frac{\Pi(c) + \fcost}{q(c)} = c + \frac{\int_c^{\overline{c}} q(x)dx + \fcost}{q(c)}$. We get $\frac{dp(c)}{dc} = -\frac{q'(c)[\Pi(c)+\fcost]}{[q(c)]^2} > 0$, i.e., $p(c)$ is increasing in $c$. We show that $q'(c) < 0$ shortly.
 
 Now we show that the direct unit tax is progressive for $c \in [\hat{c}, \overline{c})$, which is equivalent to showing that $g(c)/q(c)$ is nondecreasing in $c$, i.e., 
\begin{equation}
    P(q(c)) - c  - \frac{\fcost}{q(c)} - \frac{\Pi(c)}{q(c)} \text{ increases in } c, \label{monotone}
\end{equation}
for $c \in [\hat{c}, \overline{c})$. We show that nondecreasing, log-concave $f$, and strictly log-concave $P^{-1}$ constitute sufficient conditions for the desired monotonicity. For $c \in [\hat{c}, \overline{c})$, condition \eqref{ctrl:lambda_q} implies \[
 P(q(c)) = c + (1-\alpha)\frac{F(c)}{f(c)} + \frac{\Gamma(\hat{c}) + \beta}{f(c)} = \phi(c; \Gamma(\hat{c})),
 \] 
 which is strictly increasing in $c$ because $f(c)$ is nonincreasing, implying $q(c) = P^{-1}[\phi(c; \Gamma(\hat{c}))]$ strictly decreases over $[\hat{c}, \overline{c})$.
 The expression in \eqref{monotone} becomes
\[
 (1-\alpha)\frac{F(c)}{f(c)} + \frac{\Gamma(\hat{c}) + \beta}{f(c)} - \frac{\int_c^{\overline{c}}P^{-1}[\phi(x; \Gamma(\hat{c}))]\, dx + \fcost}{P^{-1}[\phi(c; \Gamma(\hat{c}))]}.
\]
The first two terms $(1-\alpha)\frac{F(c)}{f(c)} + \frac{\Gamma(\hat{c}) + \beta}{f(c)}$ is nondecreasing as $f(c)$ is nonincreasing.
What remains is to show $\frac{\int_c^1P^{-1}[\phi(x; \Gamma(\hat{c}))]\, dx + \fcost }{P^{-1}[\phi(c; \Gamma(\hat{c}))]}$ is decreasing. Denote the demand function to be $D:= P^{-1}$, and define the composite function $H := D \circ \phi_{\Gamma(\hat{c})}$. It is equivalent to show $\frac{\int_c^1H(x)\, dx + \fcost}{H(c)}$ is decreasing. First, consider the case $\fcost = 0$, then $\frac{\int_c^1H(x)\, dx}{H(c)}$ is strictly decreasing if $H$ is strictly log-concave.\footnote{
See Theorem 6 in \cite{BagBerg05Logconcave}.
}
Next, we show that a set of sufficient conditions for $H$ to be strictly log-concave is that $D$ is strictly log-concave, as well as $f$ being log-concave and nonincreasing. 
By Theorem 3 in \cite{ZouLogConcave} (p.\ 5), a set of sufficient conditions for the composite function $H = D \circ \phi_{\Gamma(\hat{c})}$ to be strictly log-concave is (i) $D$ is strictly log-concave and decreasing; (ii) $\phi_{\Gamma(\hat{c})}$ is monotone and  convex. We only need to show that $f$ being log-concave and decreasing leads to convex $\phi_{\Gamma(\hat{c})}$, which is true by directly verifying the second-order derivative: \[
\phi_{\Gamma(\hat{c})}''(x) = (1-\alpha)\frac{-f' \cdot f^2 +  F\cdot [2(f')^2 -f''\cdot f]}{f^3} + [\Gamma(\hat{c}) + \beta] \cdot \frac{2(f')^2-f''\cdot f  }{f^3}
\]
Log-concavity of $f$ implies $f\cdot f'' - (f')^2\leq 0$, ensuring $\phi_{\Gamma(\hat{c})}''(x)  \geq 0$. We conclude that $H$ is strictly log-concave. On the other hand, for $k > 0$, differentiating $\frac{\int_c^1H(x)\, dx + \fcost}{H(c)}$ yields \[
-1 + \frac{-H'(c)}{H(c)} \left[\frac{\int_c^1H(x)\, dx}{H(c)}+ \frac{k}{H(c)}\right].
\]
For small $k$, we show that above expression is negative, which is equivalent to $(-H'(c)) \cdot \fcost + (-H'(c)) \cdot \int_c^1H(x)\, dx - H^2(c) < 0$.
$H(c)$ being strictly log-concave implies $\int_c^{\overline{c}} H(x)\,dx$ is strictly log-concave, leading to $(-H'(c)) \cdot \int_c^1H(x)\, dx  - H^2(c) < 0, \forall c \in  [\hat{c}, \overline{c})$. Since $-H'(c) \geq 0$ is finite, there exists $\overline{k} > 0$ such that for all $k < \overline{k}$, the above inequality hold. 
To conclude, the sufficient conditions in Proposition \ref{prop_main} ensure the progressive tax condition \eqref{monotone}. This finishes the proof of Proposition \ref{prop_main}. \hfill \qedsymbol

\paragraph{Characterization of the Optimal Mechanism} In this part, we describe our strategy to find the optimal mechanism $(q, \Pi)$, and back out the firm price $p$ as well as unit tax policy $\tau$, given that the sufficient conditions in Proposition \ref{prop_main} hold. As described at the beginning of this appendix, there are two steps. First, fix $\overline{c} \in [0, \overline{c}_{LF}]$, we can pin down the truncated optimal mechanism. Then, the globally optimal mechanism is solved by optimizing across $\overline{c}$. 

In the first step, suppose $\overline{c} \in [0, \overline{c}_{LF}]$ is fixed. If $\overline{c} = 0$, the welfare is $0$, which is dominated by the always-feasible laissez-faire allocation. Now consider some fixed $\overline{c} \in(0, \overline{c}_{LF}]$. We show that the value of $\overline{c}$ corresponds to the terminal value of the cumulative multiplier function, which pins down $\hat{c}$. Note that for $c\in [\hat{c}, \overline{c}]$, $P(q(c)) = \phi(c; \Gamma(\hat{c}))$. By Lemma \ref{lem:IR_jump}, $q(\overline{c})$ is pinned down by $q(\overline{c})[P(q(\overline{c})) - \overline{c}] = \fcost$. Denote the solution as $q(\overline{c}; \fcost)$. Rearranging $\phi(c; \Gamma(\hat{c}))$ yields \[
\Gamma(\hat{c}) + \beta  = [P(q(\overline{c}; \fcost)) - \overline{c}] f(\overline{c}) - (1-\alpha)F(\overline{c}).
\]
In other words, we can rewrite the optimal policy as a function $\overline{c}$ instead of $\hat{c}$. Slightly abusing notation, write \[
\phi(c; \overline{c}):= c+ (1-\alpha)\frac{F(c)}{f(c)} + \frac{(P(q(\overline{c}; \fcost))-\overline{c})f(\overline{c}) - (1-\alpha)F(\overline{c})}{f(c)},
\]
then $q(c) = P^{-1}(\phi(c; \overline{c}))$, which only depend on exogenous parameters and the fixed choice of $\overline{c}$. This finishes the characterization of $q(c)$ for $c \in [\hat{c}, \overline{c}]$.

Now we turn to the policy for $c < \hat{c}$.
 Recall that the cutoff $\hat{c}$ is such that $P(q(\hat{c})) - \hat{c} = \frac{\Pi(\hat{c})+\fcost}{q(\hat{c})}$, and that $g(\hat{c}) = 0, g(\hat{c}+\epsilon) > 0$ for all small $\epsilon > 0$. By Lemma \ref{lem:no-subsidy-behavior}, $q(\hat{c}) > q_{LF}(\hat{c})$. Since the no-subsidy constraint \eqref{controlFS} binds for all $c\in [0,\hat{c}]$, the policy $q(c)$ is either flat for laissez-faire. Because $q(\hat{c}) > q_{LF}(\hat{c})$, it must be that $q(\hat{c}_-)' = 0$. By continuity (Lemma \ref{lem:continuous}), $q(c)$ is flat as $c$ decreases. If there exists $c_L \in (0, \hat{c})$ such that $q(\hat{c}) = q_{LF}(c_L)$, then $q(c) \equiv q(\hat{c})$ for $c \in [c_L, \hat{c}]$, and $q(c) = q_{LF}(c)$ for $c\in [0, c_L]$.
Otherwise, $q(c) \equiv q(\hat{c})$ for $c \in [0, \hat{c}]$.
To see why, suppose $c_L > 0$. continuity implies that the policy can at most split into three segments: flat, laissez-faire, and flat. We are to show that the optimal policy can only exhibit two segments: laissez-faire and flat. By condition \eqref{ctrl:lambda_q}, $\lambda_q(0_+) = -\beta \frac{\partial g(q, \Pi, c)}{\partial q} = -\beta (P(0) + qP'(0))\geq 0$, which requires $q(0) \geq q_{LF}(0)$ as long as $\beta > 0$. If $\beta = 0$, then $\lambda_q(0_+) = 0$, while $\lambda_q'(0) = -P(q(0))f(0) - \gamma(0) [P(0) + qP'(0)] < 0$ 
if $q(0) < q_{LF}(0)$, a contradiction to the requirement that $\lambda_q(c) \geq 0$ for all $c$. Therefore, it must be that $q(0) \geq q_{LF}(0)$, i.e., there is no flat segment of $q(c)$ preceding laissez-faire segment. We summarize the constrained optimal demand policy as follows: \[
q(c; \overline{c}) = \begin{cases}
    q_{LF}(c), & c \in [0, \max\{c_L,0\})\\
   P^{-1}(\phi(\hat{c}; \overline{c})), & c \in [\max\{c_L,0\}, \hat{c})\\
    P^{-1}(\phi(c; \overline{c})),& c\in  [\hat{c},\overline{c}]\\
     0,& c \in (\overline{c}, 1]
\end{cases}
\]
The profit of the firm is \[
\Pi(c; \overline{c}) = \int_c^1q(x; \overline{c})\, dx.
\]
Plugging $(q(c; \overline{c}) , \Pi(c; \overline{c}) )$ into the objective function, the expected welfare is 
 \[
 \mathcal{W}(\overline{c}) := \int_0^1 W(q(c; \overline{c}),\Pi(c; \overline{c}), c)f(c)\, dc
\]
The global optimization becomes a single variable problem\[
\max_{\overline{c} \in [0,\overline{c}_{LF}]} \mathcal{W}(\overline{c}).
\]
Since $\mathcal{W}(\overline{c})$ is continuously differentiable in $[0,\overline{c}_{LF}]$, by extreme value theorem, it admits an optimal solution $\overline{c}^*$. Ignoring the constraints on $\overline{c}$, this problem has a first-order necessary condition (equivalent to the transversality condition with a free terminal point): $W'(\overline{c}) = 0$. Despite that there is no closed-form $\overline{c}$ without adding more structure to the distributions,\footnote{
We have full characterization for the linear demand - uniform cost case in Appendix \ref{apx:linear-uniform}.
} standard comparison among the candidates of maxima (i.e., among the endpoints and the points characterized by the first-order necessary condition) yield an optimal solution $\overline{c}^*$. Using the Taxation Principle, we can construct a unit tax policy that implements the optimal direct mechanism.
We summarize the optimal mechanism in Lemma \ref{lem:main}.
\begin{lem} \label{lem:main}
Suppose that condition \eqref{lf_optimal} is violated, that $f$ is nonincreasing and log-concave, and that $P^{-1}$ is strictly log-concave. 
Then there exists some $\overline{k} > 0$, for each $\fcost \in [0, \overline{k})$,
the optimal regulation implements a demand schedule and a pricing strategy, $(q^*,p^*)$,  partitioned in at most four segments with respect to cost cutoffs $0 \leq c_{L}\leq \hat{c}\leq \overline{c} \leq 1$ and a benchmark price $\hat{p}> 0$.\begin{itemize}
    \item[(i)]  For $c\in [0, c_{L})$, the regulation induces laissez-faire demand and prices, i.e., $q^*(c) = q_{LF}(c)$ and $p^*(c) = P(q_{LF}(c))$.
    
    \item[(ii)] For $c\in [c_{L}, \hat{c})$, the regulation induces bunching at the benchmark price, i.e., $p^*(c) \equiv \hat{p}$ and $q^*(c) \equiv P^{-1}(\hat{p})$.
    
    \item[(iii)] For $c\in [\hat{c}, \overline{c}]$, the regulation induces a demand of $q^*(c) = P^{-1} \left(\phi(c; \overline{c})\right)$ and a price of $p^*(c) = c + \frac{\int_c^{\overline{c}} q^*(x)dx+ \fcost}{q^*(c)}$.
     In particular, the \emph{benchmark price} is $\hat{p} : = p^*(\hat{c})$.
    \item[(iv)] For $c\in (\overline{c},1]$, the firm is excluded from the market, i.e., $q^*(c) = 0$. 
\end{itemize}
Finally, the following unit tax policy implements the optimal mechanism: \[
\tau^*(p) = \begin{cases}
    0, & p \leq \hat{p}\\
    (P\circ q^*)[p^{*-1}(p)] - p, & p > \hat{p}
\end{cases}
\]
\end{lem}

\newpage
\bibliographystyle{ecta}
\bibliography{reference}

\pagebreak

\section{Supplementary Materials}

\subsection{Closed-form solutions to the linear-uniform case}
\label{apx:linear-uniform}

We obtain closed form characterization regulation policy for the cases with linear demand curve and uniform cost distribution $F(c) = c \in [0,1]$. For expositional simplicity, set $\fcost = 0$, so the demand policy $q(c)$ is continuous on $[0,1]$ with $q(\overline{c}) = 0$. Let the inverse demand be $P(q) = A - Bq$ ($A \leq 1, A \leq 2B$)
Following the derivation in Appendix \ref{apx:main}, there are two steps: (1) Given a fixed $\overline{c}\in [0,A]$, derive the optimal mechanism candidate $(q(c; \hat{c}(\overline{c})), p(c; \hat{c}(\overline{c})))$ as a function of $\overline{c}$; (2) optimize over $\overline{c}$.

Starting with a fixed $\overline{c}$, we have \[
P(q(c;\hat{c}(\overline{c}))) = A - (2-\alpha)(\overline{c} - c) \text{ for } c\in [\hat{c}(\overline{c}), \overline{c}]
\]
because $c + (1-\alpha)F(c)/f(c) = (2-\alpha)c$. This implies \[
q(c;\hat{c}(\overline{c}))  = \frac{1}{B} (2-\alpha)(\overline{c} - c) \text{ for } c\in [\hat{c}(\overline{c}), \overline{c}]
\]
As a result, the value function (profit) is \[
\Pi(c;\hat{c}(\overline{c})) = \frac{1}{B} \int_c^{\overline{c}} q(x;\hat{c}(\overline{c}))dx = \frac{1}{2B} (2-\alpha)(\overline{c} - c)^2 \text{ for } c\in [\hat{c}(\overline{c}), \overline{c}]
\]
We choose $\hat{c}$ so that the no-subsidy constraint bind. That is, $\hat{c}(\overline{c})$ solves \[
\Pi(c;\hat{c}(\overline{c})) = q(c;\hat{c}(\overline{c}))P(q(c;\hat{c}(\overline{c}))) - cq(c;\hat{c}(\overline{c}))
\]
with respect to $c$, which leads to \[
\hat{c}(\overline{c}) = \frac{(2(2-\alpha)+1)\overline{c} - 2A}{2(2-\alpha) - 1}
\]
Note that depending on $A$, the choice of $\overline{c}$ needs to be large enough to guarantee $\hat{c}(\overline{c}) > 0$. Plugging $\hat{c}(\overline{c})$ into $q(c; \hat{c}(\overline{c}))$, we get \[
q(\hat{c}(\overline{c})) = \frac{2(2-\alpha)}{B} \frac{A - \overline{c}}{2(2-\alpha)-1}
\] 
The welfare as a function of $\overline{c}$ is calculated by \[
\mathcal{W}(\overline{c}) = \int_0^{1} W(q(c; \hat{c}(\overline{c})), \Pi(c; \hat{c}(\overline{c})))dc
\]
where \[
q(c;  \hat{c}(\overline{c})) = \begin{cases}
    \frac{2(2-\alpha)}{B} \frac{A - \overline{c}}{2(2-\alpha)-1},& c \in [0, \hat{c}(\overline{c})]\\
    \frac{1}{B}(2-\alpha)[\overline{c} - c]^+ & c \in [\hat{c}(\overline{c}), 1]
\end{cases}
\]
and \[
\Pi(c;  \hat{c}(\overline{c})) = \begin{cases}
    \frac{1}{2B}(2-\alpha)(\overline{c} - \hat{c}(\overline{c}))^2 + [\hat{c}(\overline{c}) - c]q(\hat{c}(\overline{c});\hat{c}(\overline{c})),& c\in [0, \hat{c}(\overline{c})]\\
    \frac{1}{2B}(2-\alpha)(\overline{c} - c)^2 \mathbbm{1}_{c \leq \overline{c}}, & c\in [\hat{c}(\overline{c}),1]
\end{cases}
\]
Evaluating the integral of $\mathcal{W}(\overline{c})$, we get a smooth function of $\overline{c}$. Optimizing with respect to $\overline{c}$, we get $\overline{c}^*\in [0,A]$. Substituting $\overline{c}$ with $\overline{c}^*$ in $q(c; \hat{c}(\overline{c}))$ and $\Pi(c; \hat{c}(\overline{c}))$, we get the optimal demand policy. The induced pricing strategy is given by $p(c; \hat{c}(\overline{c})) =c +  \Pi(c; \hat{c}(\overline{c}))/q(c; \hat{c}(\overline{c}))$ for $c < \overline{c}$. 

To be concrete (and to avoid case discussions without getting new insights), consider the case $A = B = 1$. Tedious but straightforward calculation leads to \[
\mathcal{W}(\overline{c};\alpha) = \frac{(1-\overline{c})(2-\alpha)(8\alpha - 10 - \overline{c}(2-\alpha) [-28 + 24\alpha + (23- 12\alpha(3-\alpha))\overline{c}])}{3(3 - 2\alpha)^3}
\]
i.e., a cubic function in $\overline{c}$. The optimal $\overline{c}(\alpha)$ has a complicated closed-form. Instead of reporting the general solution, we look at two extreme cases discussed in the main text, i.e., $\alpha = 0$ or $1$.
When $\alpha =0$, we have \[
\mathcal{W}(\overline{c};0) = \frac{4}{81}(1-\overline{c})^2(23\overline{c}-5) 
\]
which attains maximum at $\overline{c}^* = \frac{11}{23}$, which implies $\hat{c}(\overline{c}^*) = \frac{3}{23}$, the pricing strategy and demand schedule given by \[
p^*(c) = \begin{cases}
   \frac{7}{23},& \text{ if } c\in [0, \frac{3}{23}]\\
    \frac{1}{2}c + \frac{11}{46}, & \text{ if } c\in (\frac{3}{23},  \frac{11}{23}]\\
    1, & \text{ if } c\in (\frac{11}{23},1]
    \end{cases},\;
    q^*(c) = \begin{cases}
    \frac{16}{23},& \text{ if } c\in [0, \frac{3}{23}]\\
    \frac{22}{23} - 2c, & \text{ if } c\in (\frac{3}{23},  \frac{11}{23}]\\
    0, & \text{ if } c\in (\frac{11}{23},1]
    \end{cases}
\]
as well as the consumer price\[
P(q^*(c)) = 1 - q^*(c) = \begin{cases}
   \frac{7}{23},& \text{ if } c\in [0, \frac{3}{23}]\\
    2c + \frac{1}{23}, & \text{ if } c\in (\frac{3}{23},  \frac{11}{23}]\\
    1, & \text{ if } c\in (\frac{11}{23},1]
    \end{cases}
\]

Therefore, a version of unit tax policy that implements the optimal demand and price is
\[
\tau(p) = P(q^*(p^{*-1}(p))) - p = \begin{cases}
    0& \text{ if } p  \leq \frac{7}{23} \\
    3p - \frac{21}{23} & \text{ if } p \in (\frac{7}{23}, \frac{11}{23}] \\
    \infty, & \text{ if } p >  \frac{11}{23}
    \end{cases}
\]

On the other hand, when $\alpha = 1$, \[
\mathcal{W}(\overline{c}; 1) = \frac{1}{3}(1-\overline{c})(\overline{c}^2 + 4\overline{c} - 2)
\]
which attains maximum at $\overline{c}^* = \sqrt{3} - 1$, which implies $\hat{c}(\overline{c}^*) = 3\sqrt{3} - 5$, the pricing strategy and demand schedule given by \[
p^*(c) = \begin{cases}
   2\sqrt{3} - 3,& \text{ if } c\in [0, 3\sqrt{3} - 5]\\
    \frac{1}{2}c + \frac{\sqrt{3}-1}{2}, & \text{ if } c\in (3\sqrt{3} - 5,  \sqrt{3} - 1]\\
    1, & \text{ if } c\in (\sqrt{3} - 1,1]
    \end{cases},\;
    q^*(c) = \begin{cases}
    4 - 2\sqrt{3},& \text{ if } c\in [0, 3\sqrt{3} - 5]\\
    \sqrt{3} - 1 - c, & \text{ if } c\in (3\sqrt{3} - 5,  \sqrt{3} - 1]\\
    0, & \text{ if } c\in (\sqrt{3} - 1,1]
    \end{cases}
\]
as well as the consumer price\[
P(q^*(c)) = 1 - q^*(c) = \begin{cases}
   2\sqrt{3} - 3,& \text{ if } c\in [0, 3\sqrt{3} - 5]\\
    2 - \sqrt{3} + c, & \text{ if } c\in (3\sqrt{3} - 5,  \sqrt{3} - 1]\\
    1, & \text{ if } c\in (\sqrt{3} - 1,1]
    \end{cases}
\]

Therefore, a version of unit tax policy that implements the optimal demand and price is
\[
\tau(p) = P(q^*(p^{*-1}(p))) - p = \begin{cases}
    0& \text{ if } p  \leq 3\sqrt{3} - 5 \\
    2p + 3 - 2\sqrt{3} & \text{ if } p \in (3\sqrt{3} - 5, \sqrt{3} - 1] \\
    \infty, & \text{ if } p > \sqrt{3} - 1
    \end{cases}
\]

\subsection{A numerical example where laissez-faire prices are induced under progressive price cap}\label{apx:normal}
 Suppose that the firm's cost $c\in[0,1]$ follows a  (truncated) normal distribution\footnote{
This is a distribution generated by $N(0.5, 0.01)$ conditional on assuming values in $[0,1]$. For illustration purpose, we use the normal approximation (unconditional normal distribution $F\sim N(0.5, 0.01)$) for this example with small $\sigma$. Indeed, for $N(0.5, 0.01)$, the tail probabilities are ignorable as $1-\Pr(c \geq 1) = \Pr(c\leq 0) < 10^{-6}$. 
} with mean $0.5$ and variance $0.01$. For concreteness, fix the welfare weight $\alpha = 0.1$ and fixed cost $\fcost = 0$.
This example demonstrates a case where most cost types set price near the benchmark, with extremely low-cost types enjoying laissez-faire pricing and extremely high-cost types excluded from the market.
The optimal regulation and induced pricing strategy are visualized in Figure \ref{fig:normal}. 

\begin{figure}[htbp]
    \begin{center}
               \begin{tikzpicture}[scale=5.5]
        \footnotesize
            \draw [->, thick,gray] (0,0) -- (0,1.2);
            \draw [->, thick,gray] (0,0) -- (1.2,0);

             \draw[-, very thick, dotted] (0,1) -- (1,0);

             \draw[-, nberblue, thick, dashed] (0, 0.552) -- (1-0.552, 0.552);
             \node at (-0.05, 0.552) {$\hat{p}$};

             \node at (0.16, 0.65) {taxation};
             \node at (0.18, 0.3) {delegation};

             \draw[-, ultra thick, nberblue] (0,0.609) -- (0, 1); 
             \draw[-, ultra thick, nberblue] (0,0.609) [out=-10 , in = 155] to (1-0.552, 0.552);
             \draw[-, ultra thick, nberblue] (1,0) -- (1-0.552, 0.552);
             \node at (-0.1, 0.8) {{\color{nberblue}$P_\tau(q)$}};

             \node at (-0.05, 1) {$1$};
             \node at (1, -0.05) {$1$};
            
            \node at (-0.05, 1.2) {$p$};
           
            \node at (1.22, -0.05) {$q$};
            \node at (0.15, 0.95) {$P(q)$};

            \draw [->, thick,gray] (1.5,0) -- (1.5,1.2);
            \draw [->, thick,gray] (1.5,0) -- (2.7,0);

            \draw [-, gray, dashed] (1.5,0) -- (2.5,1);

            \node at (1.45, 1) {$1$};
            \node at (1.44, 0.48) {$0.5$};
            
            \node at (1.45, 1.2) {$p$};
          
             \draw [-, dashed, gray] (1.5,1) -- (2.5,1);

            \fill[uncblue!20] (1.5,0.5) -- (1.5 + 0.104, 0.5515) -- (1.5 + 0.104, 1) -- (1.5, 1); 
             \fill[uncblue!20]  (1.5 + 0.104, 0.5515) -- (1.5 + 0.46, 0.5515) -- (1.5 + 0.46, 1) -- (1.5 + 0.104, 1); 
               \fill[uncblue!20] (1.5 + 0.46, 0.5515) [out=55, in=265] to (1.5 + 0.609, 1) -- (1.5 + 0.46, 1); 
              
            \draw [-, nberblue!60, very thick, dashed] (0.46+1.5,0.5515) [out=55, in=265] to (0.609+1.5,1);

            \node at (2.1, 1.05) {${\color{nberblue}P[q^*(c)]}$};
            
            \node at (1.5 + 0.105, -0.05) {$c_{L}$};
            \node at (1.5 + 0.46, -0.05) {$\hat{c}$};
            \node at (2.72, -0.05) {$c$};
            \node at (1.5+0.609, -0.05) {$\overline{c}$};

            \draw[gray, dashed] (1.5+0.609, 0) -- (1.5+0.609, 0.609);

            \draw [-, ultra thick] (1.5 + 0.104, 0.5515) -- (0.46+1.5, 0.5515);
            \draw [-, ultra thick] (0.46+1.5, 0.5515) [out=10, in=220] to (0.609+1.5, 0.609);
            \draw [-, gray!70, dashed] (1.5 + 0.104, 0.5515) -- (1.5, 0.5515);
            \node at (2.2, 0.6) {$p^*(\cdot)$};
             \draw [-, dashed, gray] (1.5 +0.104, 0) -- (1.5 + 0.104, 0.5515);
            \draw [-, dashed, gray] (1.5 +0.46, 0) -- (1.5 + 0.46, 0.5515);
            \draw[-, ultra thick] (1.5,0.5) -- (1.5 + 0.104, 0.5515);

            \draw [-, thick, purple, dotted] (1.5, 0.5) -- (2.5, 1);
            \node at (2.6, 0.95) {${\color{pennred}P[q_{LF}(\cdot)]}$};
           
            \draw [-, thick, purple!60, dotted] (1.5, 0.5) -- (2.5, 1);
           
             \node at (1.42, 0.552) {$\hat{p}$};
            
            \fill[pennred, opacity = 0.2] (1.5+0.46, 0.5515) [out=55, in=265] to (1.5 + 0.609, 1) --(1.5 + 0.609, 0.609) [out = 220, in = 10] to (1.5 + 0.46, 0.5515); 
            \node at (0.75+1.5, 0.4) {\textbf{tax}}; 
            \draw[-] (2.05, 0.65) -- (2.18, 0.43);
        \end{tikzpicture}
    \end{center}
    \caption{\footnotesize Firm's inverse demand function $P_\tau(q)$ (left) and the induced pricing strategy $p^*(c)$ (right). $P(q) = 1 - q$, $c \sim N(0.5, 0.01)$, $\alpha = 0.1$.}
    \label{fig:normal}
\end{figure}
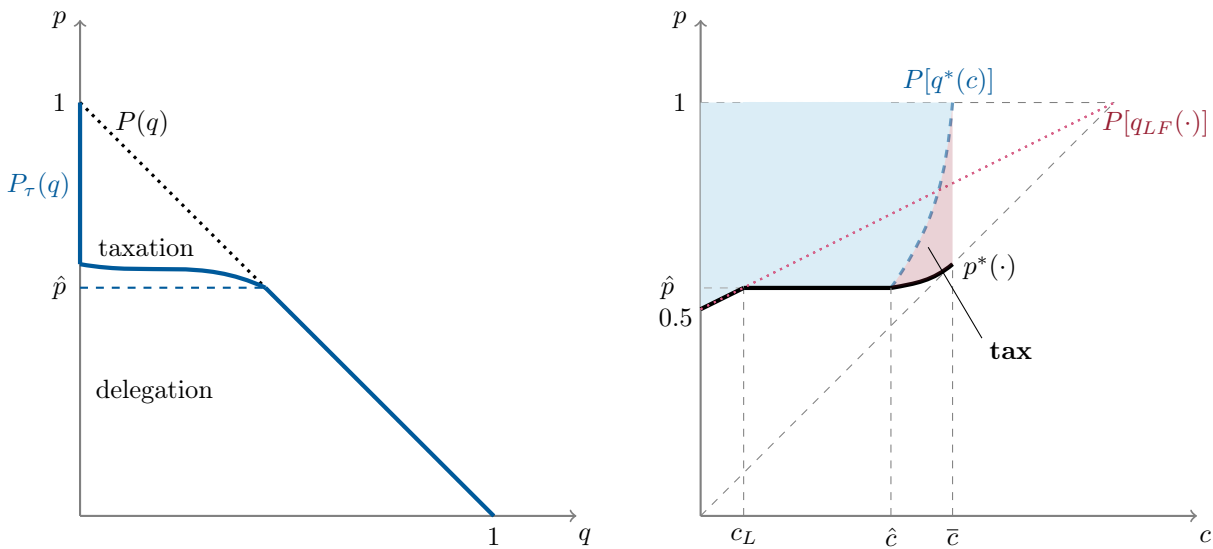

\end{document}